\documentclass[lettersize,journal]{IEEEtran}
\usepackage{amsthm}
\usepackage{amsmath,amsfonts}
\usepackage{algorithmic}
\usepackage{algorithm}
\usepackage{amssymb}
\usepackage{array}
\usepackage[caption=false,font=normalsize,labelfont=sf,textfont=sf]{subfig}
\usepackage{textcomp}
\usepackage{stfloats}
\usepackage{url}
\usepackage{verbatim}
\usepackage{graphicx}
\usepackage{cite}
\usepackage{cases}
\usepackage{hyperref}
\usepackage{bm}
\usepackage{color}
\usepackage{gensymb}
\usepackage{multirow}
\usepackage{booktabs}
\usepackage{makecell}
\usepackage{threeparttable}
\usepackage{utfsym}
\DeclareMathOperator*{\argmax}{arg\,max}

\hypersetup{colorlinks=true,
	linkcolor=blue,
	citecolor=blue,      
	urlcolor=black,
}

\newtheoremstyle{mylemma}
{}{}                 
{\normalfont}        
{}                   
{\bfseries}         
{.}                  
{ }                  
{\thmname{#1}\thmnumber{ #2}\thmnote{ (#3)}} 

\theoremstyle{mylemma}

\newtheorem{remark}{Remark}

\newtheorem{proposition}{Proposition}

\begin{document}

\title{Fast Beam-Brainstorm: Few-Step Generative Site-Specific Beamforming with Flexible Probing}

\author{Zihao Zhou,~\IEEEmembership{Graduate Student Member,~IEEE}, Zhaolin Wang,~\IEEEmembership{Member,~IEEE}, and Yuanwei Liu,~\IEEEmembership{Fellow,~IEEE}
\thanks{The authors are with the Department of Electrical and Computer Engineering, The University of Hong Kong, Hong Kong (e-mail: eezihaozhou@connect.hku.hk,zhaolin.wang@hku.hk,yuanwei@hku.hk)}}

\maketitle

\begin{abstract}
A novel generative site-specific beamforming (GenSSBF) approach, termed fast beam-brainstorm (FBBS), is proposed to address two practical bottlenecks in GenSSBF, namely, slow beam generation and fixed channel probing budgets. To accelerate beam generation, FBBS employs a two-stage training strategy that learns an average velocity field, instead of an instantaneous one, to guide the beam generative process, reducing the generation steps to few or even one. To support flexible probing, a stochastic masking mechanism and a beam index-aware masked condition encoder are proposed, enabling a single trained model to handle variable-length channel observations, i.e., reference signal received power (RSRP), without retraining. Therefore, FBBS achieves the fast generation of high-fidelity communication beams from coarse and variable-length channel probing feedback. Simulations on ray-tracing datasets confirm that FBBS achieves a favorable trade-off between beamforming gain, probing overhead and beam generation speed. Moreover, FBBS demonstrates robustness to noisy and outdated RSRP measurements, and meaningful cross-scenario zero-shot generalization.
\end{abstract}

\vspace{-1.2em}
\section{Introduction}
\IEEEPARstart{M}{ultiple}-antenna systems rely critically on beamforming to leverage spatial degrees of freedom (DoF) for enhanced coverage and capacity \cite{bjornson2023twenty}, particularly at high-frequency bands where directional gain is essential to combat severe propagation and penetration losses\cite{xiao2017millimeter, kutty2016beamforming}. While optimal beamforming theoretically requires perfect channel state information (CSI), real-time acquisition is fundamentally hampered by complexity, overhead, and mobility.

In the current fifth generation new radio (5G NR), beamforming relies on a ``probing-measuring-reporting'' framework over pre-defined codebooks (e.g., discrete Fourier transform (DFT) or oversampled-DFT (O-DFT) codebook). The base-station (BS) transmits reference signals (RSs) such as synchronization signal blocks (SSBs) and CSI-RSs, and the UE reports quality metrics or beam indices \cite{dahlman20185g}. Various sweeping methods build on this paradigm \cite{he2015suboptimal,xiao2016hierarchical,qi2020hierarchical,chen2018dynamic}. While robust, simple, and applicable to both time-division duplexing (TDD) and frequency-division duplexing (FDD) without full CSI, this exhaustive-search approach suffers three fundamental limitations. 1) The codebooks are designed via mathematical orthogonality rather than adapting to the specific propagation characteristics of the deployment site (e.g., dominant angle distributions and scatterer geometries). 2) Single-lobe DFT/O-DFT beams limit gain especially in non-line-of-sight (NLoS) scenarios, and the grid-mismatch problem persists \cite{abdallah2025explainable}. 3) The beam search space and the associated pilot overhead grow with array size, becoming prohibitive for XL-MIMO with hundreds to thousands of antennas \cite{liu2023near}.

The aforementioned limitations motivate \emph{site-specific beamforming (SSBF)}, which exploits prior knowledge of the propagation environment to reduce online sweeping burden\cite{heng2024site}. In recent years, deep learning (DL)-assisted SSBF has attracted extensive research attention. A typical solution is to parameterize probing codebooks via complex-valued neural networks (NNs), where beamforming vectors are instantiated by learned NN weights \cite{alrabeia2022neural}. Following this NN-as-codebook paradigm, an end-to-end learning framework is proposed to jointly optimize site-specific probing codebooks and beam selectors \cite{heng2022learning}. This framework is further extended to hierarchical beam alignment in \cite{yang2024hierarchical}. Beyond DL-based solutions, reinforcement learning (RL) was adopted to learn codebooks online where the agent receives reward signals from the environment to improve the beam pattern \cite{zhang2022reinforcement, abdallah2024multi}. Although a more ``meaningful'' codebook can be obtained via the aforementioned approaches, the ``grid mismatch'' issue still remains. Therefore, in \cite{heng2024grid}, a grid-free SSBF framework was proposed where the probing codebooks and a multi-layer perceptron (MLP)-based beam synthesizers were jointly trained. Similarly, the application of NN for beam synthesis is seen in \cite{abdallah2025explainable} for hybrid beamforming and in \cite{li2025site} for joint transmit- and receive beamforming in full-duplex system.

However, it is noticed that most existing SSBF pipelines are fundamentally using \emph{discriminative models}: they map input (i.e., coarse CSI measurements) to beam index or beamforming vector. Although discriminative models have proven effective in some domains, for SSBF, this paradigm presents the following limitations. First, conditional discriminative models often struggle to capture \textit{multimodal distributions} \cite{torralba2024foundations}. Specifically, in SSBF, such multimodality naturally arises under multipath propagation, blockage uncertainty, and noisy feedback, where different receiver positions can yield similar measurements (spatial ambiguity) \cite{wang2026generative, zhou2026beam}. From a learning perspective, standard point-estimation objectives tend to collapse this multimodality \cite{torralba2024foundations}. For example, under a minimum mean squared error (MMSE) objective, the predictor is driven toward the conditional mean, which may lie between several valid beam patterns and become physically suboptimal. Moreover, the existing DL-based SSBF solutions typically output a single prediction. When faced with ambiguous input, the model will still have one deterministic output during inference, which can lead to poor beam decisions and subsequent error propagation\cite{wang2026generative}. In addition, discriminative model-based SSBF essentially belongs to the category of ``\textit{unstructured beam prediction}''\cite{torralba2024foundations}, as element-wise regression objectives do not explicitly enforce the global phase-coherent structure of a valid beamforming vector. This limitation, however, fails to account for the fact that these elements are highly correlated and collectively determine whether the signals combine constructively or destructively at the receiver\cite{masouros2015exploiting}.

Consequently, we reformulate beamforming as a generative modeling task that aims to learn a multimodal conditional distribution rather than a deterministic mapping, giving rise to generative site-specific beamforming (GenSSBF) \cite{wang2026generative,zhou2026beam,zhao2026generative}. In our previous work, a beam-brainstorm (BBS) solution was proposed using a conditional diffusion model with compact RSRP vectors as wireless prompts, achieving near-optimal gain with substantially lower probing overhead. Meanwhile, \cite{zhao2026generative} designs a site information-maximizing (SIM) codebook whose RSRP measurements guide a flow matching (FM)-based beam generator. Despite this progress, two significant challenges remain: \textbf{\textit{how to achieve fast beam generation}}, and \textbf{\textit{how to accommodate on-demand flexible probing}}.
For the first question, our prior methods \cite{zhou2026beam,zhao2026generative} rely on iterative sampling. When the generation steps are aggressively reduced for real-time operation, errors may accumulate and degrade beam quality significantly. For the second question, the conditional generative models assume fixed-dimensional prompts, whereas practical systems frequently vary probing budget according to signaling load, latency targets, mobility, and feedback availability. A model tied to one fixed probing length/pattern is therefore difficult to deploy in adaptive networks. Motivated by the above challenges, this paper proposes \emph{fast beam-brainstorm (FBBS)}, a new GenSSBF solution targeting high-fidelity \emph{one or few-step} beam generation under \emph{flexible} CSI probing. Compared with our previous GenSSBF framework, FBBS retains the RSRP-conditioned generative beamforming formulation. However, we introduce average velocity learning for one- or few-step beam generation. Additionally, unlike existing GenSSBF methods that assume a fixed-dimensional wireless prompt, FBBS adopts stochastic masking and a beam index-aware condition encoder to support variable probing budgets using a single trained model. A systematic comparison with existing GenSSBF methods is provided in Table~\ref{tab:comparison_existing_genssbf}. The main contributions are thus summarized as follows:

\begin{table*}[t]
	\centering
	\caption{Methodological comparison of existing GenSSBF frameworks.}
	\label{tab:comparison_existing_genssbf}
	\scriptsize
	\setlength{\tabcolsep}{4pt}
	\renewcommand{\arraystretch}{1.2}
	
	\begin{tabular}{lccccc}
		\toprule
		\textbf{Method}
		& \textbf{Generative Paradigm}
		& \textbf{Training Procedure}
		& \textbf{Wireless Prompt}
		& \textbf{Flexible Probing}
		& \textbf{Online Generation Cost} \\
		\midrule
		
		BBS~\cite{zhou2026beam}
		& \makecell{Conditional reverse\\denoising}
		& \makecell{Single-stage\\denoising training}
		& \makecell{Fixed-length RSRP}
		& \scalebox{2}{\usym{2718}}
		& \makecell{Very high: iterative\\denoising (e.g., $T=1000$)} \\
		\midrule
		
		SIM-GenSSBF~\cite{zhao2026generative}
		& \makecell{Instantaneous-\\velocity-guided movement}
		& \makecell{Single-stage\\flow matching}
		& \makecell{Fixed-length RSRP}
		& \scalebox{2}{\usym{2718}}
		& \makecell{Moderate: multi-step\\ODE integration (e.g., $T=40$)} \\
		\midrule
		
		\textbf{FBBS}
		& \makecell{Average\\velocity-guided movement}
		& \makecell{FM initialization $+$\\split-consistency fine-tuning}
		& \makecell{Variable-length RSRP}
		& \scalebox{2}{\usym{2714}}
		& \makecell{\textbf{Low: one- or few-step}\\\textbf{generation}  (e.g., $T<3$)}  \\
		
		\bottomrule
	\end{tabular}
\end{table*}

\begin{itemize}
	\item We propose \underline{F}ast \underline{B}eam-\underline{B}rain\underline{s}torm (FBBS), a more practical GenSSBF solution. FBBS learns the average velocity field of the beam evolution process through a two-stage strategy. Stage I learns the instantaneous velocity field that transports samples from a prior distribution to the target beam distribution following an optimal-transport path. Stage II fine-tunes this model to an average-velocity predictor via interval-splitting consistency, enabling one or few-step fast beam generation. We further derive a terminal error bound for FBBS, which explicitly decouples three error sources--velocity approximation, prompt imperfections, and inter-step error amplification, thereby establishing the finite-horizon stability of FBBS under practical probing imperfections.
	
	\item To address the fixed-length probing constraint, we propose an online stochastic masking strategy that mixes full and partial conditions within each training batch without enlarging the dataset, thereby enabling a single model to support on-demand probing at inference with no retraining. Furthermore, the integration of beam index-aware semantics, combined with cardinality-normalized masked aggregation, ensures that the sparse RSRP measurements remain both interpretable and stable.
	
	\item Extensive simulations on accurate ray-tracing datasets covering diverse indoor/outdoor and LoS/NLoS environments validate the effectiveness of the proposed FBBS. The results demonstrate that 1) FBBS achieves a superior gain-overhead-speed tradeoff, 2) learning average velocity field for beam evolution, rather than instantaneous local dynamics, is crucial for reliable few-step beam generation, 3) the advantage of FBBS is most pronounced in the low-overhead regime, where mild brainstorming with compact RSRP measurements already captures most of the achievable performance gain, and 4) FBBS remains robust across varying probing budgets and degraded RSRP measurements.
\end{itemize}

The remainder of this paper is organized as follows. The system model and problem formulation are described in Section \ref{Section2}. In Section \ref{Section3}, the proposed FBBS solution is introduced. The datasets employed and the simulation results are provided in Section \ref{Section4}. Section \ref{Section5} concludes the paper.

\emph{Notations:} Scalars, vectors, and matrices are denoted by italic, bold lowercase, and bold uppercase letters, respectively, while calligraphic symbols denote sets or operators when needed. The superscripts $(\cdot)^{\rm T}$ and $(\cdot)^{\rm H}$ represent the transpose and Hermitian transpose, respectively. Moreover, $\mathbb{R}$ and $\mathbb{C}$ denote the real and complex fields, and $\mathbb{E}[\cdot]$ denotes expectation. For a positive integer $n$, $\{0,1\}^{n}$ denotes the set of all binary vectors of length $n$. Finally, $\odot$ denotes the Hadamard product.


\section{System Model and Problem Formulation}\label{Section2}
Consider a downlink system in which a BS equipped with an $N_t$-element uniform linear array (ULA) serves a scheduled single-antenna UE $k$. Its $L$-path channel is modeled as
\begin{equation} 
	\setlength\abovedisplayskip{3pt}
	\setlength\belowdisplayskip{3pt}
	\bm{{\rm h}}_k =\sum_{l=1}^{L}\alpha_{k,l}\bm{{\rm a}}(\phi_{k,l}),
	\vspace*{-0.5\baselineskip}
\end{equation}
\begin{equation}
	\setlength\abovedisplayskip{3pt}
	\setlength\belowdisplayskip{3pt}
	\bm{{\rm a}}(\phi) =\frac{1}{\sqrt{N_t}} \left[ 1,e^{j\frac{2\pi d}{\lambda}\sin\phi},\ldots, e^{j(N_t-1)\frac{2\pi d}{\lambda}\sin\phi} \right]^{T},
\end{equation}
where $\alpha_{k,l}$ and $\phi_{k,l}$ denote the complex gain and angle of departure of the $l$-th path, respectively, while $d$ and $\lambda$ are the antenna spacing and carrier wavelength. The BS employs analog beamforming implemented by phase shifters. Hence, the feasible beamforming set is
\begin{equation}\label{eq:feasible_beam_set}  
	\mathcal{S}\triangleq \left\{ \bm{{\rm w}}\in\mathbb{C}^{N_t\times1}: |w_n|=\frac{1}{\sqrt{N_t}},\ \forall n \right\}, 
\end{equation}
which simultaneously enforces the unit-power and constant-modulus constraints. For a given channel, the beamforming gain is $g_k(\bm{{\rm w}})=|\bm{{\rm h}}_k^{\rm H}\bm{{\rm w}}|^2$. Instead of acquiring the full CSI, the BS observes only coarse CSI through limited probing. Let $\bm{{\rm o}}_{k, Q}\in\mathbb{R}^{Q}$ denote a $Q$-dimensional
observation generated by
\vspace*{-0.5\baselineskip}
\begin{equation}
	\setlength\abovedisplayskip{3pt}
	\setlength\belowdisplayskip{3pt}
	\label{eq:coarse_observation}
	\bm{{\rm o}}_{k, Q}
	=
	\mathcal{P}_{Q}(\bm{{\rm h}}_k,\bm{\xi}_k), \qquad Q\in\mathcal{Q}_p,
\end{equation}
\vspace*{-0.2\baselineskip}
where $\mathcal{P}_{Q}(\cdot)$ represents a generic channel probing operator, $\bm{\xi}_k$ collects measurement imperfections, and $\mathcal{Q}_p$ is the set of admissible probing budgets. Due to the coarseness of the observation $\bm{{\rm o}}_{k, Q}$, multiple beamforming vectors may be plausible with the same observation. Let $q_Q^\star(\bm{{\rm w}}_k^\star \mid \bm{{\rm o}}_{k, Q}, Q)$ denote the conditional distribution of optimal beam $\bm{{\rm w}}_k^\star$. A generative beamformer maps the coarse observation and a latent variable $\bm{{\rm z}}\sim p_{\bm{{\rm z}}}$ to $\bm{{\rm w}}_k=\mathcal{G}_{\vartheta}(\bm{{\rm o}}_{k, Q},\bm{{\rm z}}, Q)\in\mathcal{S}$, which induces a conditional distribution $p_\vartheta(\bm{{\rm w}}\mid \bm{{\rm o}}_{k, Q}, Q)$. Accordingly, the general GenSSBF learning problem is formulated as
\begin{subequations}
	\setlength\abovedisplayskip{3pt}
	\setlength\belowdisplayskip{3pt}
	\label{problem_general_genssbf}
	\begin{align}
		\min_{\vartheta}\quad
		&
		\mathbb{E}_{Q, \bm{{\rm o}}_{k, Q}}
		\left[
		\mathsf{D}
		\left(
		q_Q^\star(\cdot \mid \bm{{\rm o}}_{k, Q}, Q),
		p_\vartheta(\cdot \mid \bm{{\rm o}}_{k, Q}, Q)
		\right)
		\right]
		\label{eq:general_objective}\\
		\mathrm{s.t.}\quad
		&
		\mathcal{G}_{\vartheta}(\bm{{\rm o}}_{k, Q},\bm{{\rm z}};Q)
		\in\mathcal{S},
		\quad
		Q\in\mathcal{Q}_p,
	\end{align}
\end{subequations}
where $\mathsf{D}(\cdot,\cdot)$ denotes a statistical discrepancy between two distributions. Learning the conditional distribution allows independently generated candidates to cover different beams compatible with the similar coarse observation, after which further online probing can select a high-quality candidate. In this work, the coarse observations are instantiated as RSRP measurements because they are directly compatible with the existing 5G NR beam management procedure \cite{Heng2021six}. Specifically, let $\bm{{\rm V}}=[\bm{{\rm v}}_1,\ldots,\bm{{\rm v}}_{Q_{\max}}]$ denote the DFT probing codebook. For a probing budget $Q\in \mathcal{Q}_p$, let $\Omega_Q=\{i_1,\ldots,i_Q\}\subseteq\{1,\ldots,Q_{\max}\}$ denote the selected DFT probing beam indices, where $|\Omega_Q|=Q$. When the BS transmits a unit-power reference signal using $\bm{{\rm v}}_{i}$, $i\in\Omega_Q$, the received signal is $y_{k,i}=\sqrt{P_T}\bm{{\rm h}}_k^{\rm H}\bm{{\rm v}}_i+n_{k,i}$, where $P_T$ is the transmit power, $n_{k,i}\sim\mathcal{CN}(0,\sigma^2)$. The corresponding RSRP is $r_{k,i}\triangleq |y_{k,i}|^2$. Hence, $\bm{{\rm o}}_{k, Q}=[r_{k,i_1},\ldots,r_{k,i_Q}]^{\rm T}\in\mathbb R^Q.$

\section{The Proposed Fast Beam-Brainstorm}\label{Section3}
This section details the proposed FBBS framework. First, Section \ref{section3.A} introduces the core few-step beam generation strategy. Section \ref{section3.B} then presents a conditional velocity predictor based on modified 1D-DiT model. Section \ref{section3.C} extends FBBS to variable-length RSRP conditioning, enabling a single well-trained model to flexibly accommodate different probing budgets. Finally, the complete training and inference procedures are summarized in Section \ref{section3.D}.

\vspace{-0.5\baselineskip}
\subsection{Velocity Field Learning for Beam Evolution}\label{section3.A}
Before introducing the beam generation process, we first construct data used for model training. For generality, let $\bm{{\rm c}}_k$ denote the wireless prompt of UE $k$\footnote{In this work, $\bm{{\rm c}}_k$ is instantiated as the RSRP observation $\bm{{\rm o}}_{k,Q}$, the final prompt fed into the generative model is specified in Section~\ref{section3.C}}. Denote $\bm{{\rm X}}_1^k$ as \textit{target state} for beam generation. The essence of GenSSBF is to learn the \textit{conditional distribution} of such target states given coarse CSI measurements. To construct the target state, we first transform the channel $\bm{{\rm h}}_k\in\mathbb{C}^{N_t\times 1}$ into its angular-domain representation $\bm{{\rm h}}_k^{\mathcal{A}}=\mathcal{F}(\bm{{\rm h}}_k)$, where $\mathcal{F}(\cdot)$ denotes the unitary DFT operator. This transformation exposes the sparse and energy-concentrated structure of far-field multipath channels, thereby providing a more structured target space for beam generation. We then standardize each DFT representation by fixing a deterministic phase reference and normalizing its magnitude scale.
\begin{equation}
	\setlength\abovedisplayskip{3pt}
	\setlength\belowdisplayskip{3pt}
	n_k^\star = \argmax_{n\in\{1,\ldots,N_t\}} |\bm{{\rm h}}_k^{\mathcal{A}}[n]|, \quad
	s_k = \sqrt{\frac{1}{N_t} \sum_{n=1}^{N_t} |\bm{{\rm h}}_k^{\mathcal{A}}[n]|^2},
	\label{canonical_parameters}
\end{equation}
where $\vert \cdot \vert$ denotes the amplitude of a complex number, $n_k^\star$ denotes the index of the strongest DFT coefficient and $s_k$ is the root-mean-square (RMS) magnitude. Thus, the standardized DFT representation is constructed as
\begin{equation}
	\setlength\abovedisplayskip{3pt}
	\setlength\belowdisplayskip{3pt}
	\bar{\bm{{\rm h}}}_k^{\mathcal{A}}=\frac{\bm{{\rm h}}_k^{\mathcal{A}}{\rm e}^{-j \angle \bm{{\rm h}}_k^{\mathcal{A}}[n_k^\star]}}{s_k}.
	\label{canonical_dft_channel}
\end{equation}
The phase rotation removes the global phase degree of freedom, forcing the model to focus on the relative phases among the elements of $\bm{{\rm h}}_k^{\mathcal{A}}$. The RMS normalization unifies the overall magnitude scale while preserving the relative structure among elements. Consequently, the target sample is then defined as\footnote{The generator $\mathcal G_\vartheta$ in Section \ref{Section2} is realized by generating target representation and applying the deterministic beam recovery mapping in \eqref{eq:beam_recovery}.}
\begin{equation}\label{beam_target}
	\setlength\abovedisplayskip{3pt}
	\setlength\belowdisplayskip{3pt}
	\bm{{\rm X}}_1^k = 
	\begin{bmatrix}
		\Re\{ \bar{\bm{{\rm h}}}_k^{\mathcal{A}} \}, \Im\{ \bar{\bm{{\rm h}}}_k^{\mathcal{A}} \}
	\end{bmatrix}^{\rm T} \in\mathbb{R}^{2\times N_t}
\end{equation}

Conditioned on the wireless prompt $\bm{{\rm c}}_k$, the target $\bm{{\rm X}}_1$ follows the conditional distribution $q(\bm{{\rm X}}_1\mid\bm{{\rm c}}_k)$. In the following, all distributions and velocity fields are conditioned on $\bm{{\rm c}}_k$, the conditioning argument is omitted when no ambiguity arises. Importantly, the proposed FBBS does not explicitly parameterize or estimate $q(\bm{{\rm X}}_1)$ itself. Instead, it learns a transport field that progressively maps samples drawn from a simple prior distribution $p_0(\bm{{\rm X}}_0)$ to the target distribution $q(\bm{{\rm X}}_1)$. To this end, we adopt the average velocity-guided beam evolution\cite{guo2025splitmeanflow}. Specifically, a probability density path $p_t(\bm{{\rm X}}_t)$ is constructed for $0\leq t \leq 1$ to connect the source distribution $p_0(\bm{{\rm X}}_0)$ (e.g., a standard normal distribution) and the target distribution $p_1(\bm{{\rm X}}_1)$, where $p_1(\bm{{\rm X}}_1)= q(\bm{{\rm X}}_1)$. The evolution of data samples along this path is governed by a time-dependent velocity field $\bm{{\rm v}}(\bm{{\rm X}}_t, t)$, which serves as the fundamental transport object to be learned. Therefore, the generative process is fully characterized once this velocity field is obtained. In the following, we first elaborate on the learning of the \textit{instantaneous velocity field}. Subsequently, we detail how the model is fine-tuned to effectively learn the \textit{average velocity field}, thereby enabling the generation of high-fidelity beamforming vectors in only one or a few steps.

\subsubsection{Learning Instantaneous Velocity Field}
As discussed above, to transport the source distribution $p_0(\bm{{\rm X}}_0)$ to the final distribution $p_1(\bm{{\rm X}}_1)$, we need to learn the underlying velocity field. To this end, we employ a neural network parameterized by $\vartheta$ to model the instantaneous velocity field $\bm{{\rm v}}^\vartheta(\bm{{\rm \hat{X}}}_t, t)$, which dictates the instantaneous motion of each predicted sample $\bm{{\rm \hat{X}}}_t$ along the probability density path. More particularly, $\bm{{\rm X}}_t$ denotes a state on the prescribed transport path, whereas $\bm{{\rm \hat{X}}}_t$ denotes a state generated by the model. This formulation naturally defines the trajectory of each data sample via the ordinary differential equation (ODE) as
\begin{equation}
	\setlength\abovedisplayskip{3pt}
	\setlength\belowdisplayskip{3pt}
	\frac{{\rm d}}{{\rm d}t}\bm{{\rm \hat{X}}}_t = \bm{{\rm v}}^\vartheta(\bm{{\rm \hat{X}}}_t, t), \bm{{\rm \hat{X}}}_0 = \bm{{\rm X}}_0.
\end{equation}
Once the instantaneous velocity field is learned, generation starts from a prior state $\bm{{\rm X}}_0 \sim p_0(\bm{{\rm X}}_0)$ at $t=0$, and the learned field is integrated from $t=0$ to $t=1$ to obtain the generated sample. Crucially, this deterministic ODE-based sampling is efficient, leading to significantly faster inference compared with conventional diffusion models that rely on stochastic differential equations (SDEs), which typically need hundreds or thousands of iterative steps to produce a single sample.

However, learning $\bm{{\rm v}}(\bm{{\rm X}}_t, t)$ directly is an ill-posed problem because there is no available target for this instantaneous velocity field. In fact, there are infinitely many paths to transport $p_0(\bm{{\rm X}}_0)$ to $p_1(\bm{{\rm X}}_1)$, each corresponding to a different velocity field. To make the optimization problem well-posed, we must first specify how the probability density evolves over time, i.e., we need to define a specific probability density path $p_t(\bm{{\rm X}}_t), t\in [0, 1]$ that connects $p_0(\bm{{\rm X}}_0)$ and $p_1(\bm{{\rm X}}_1)$. Therefore, the problem reduces to \textit{designing a tractable probability density path}, and then \textit{learning the corresponding instantaneous velocity field}. For the first stage, a tractable choice is the conditional optimal-transport path, which defines the trajectory of each sample via a linear interpolation between the source sample $\bm{{\rm X}}_0\sim p_0(\bm{{\rm X}}_0)$ and the final sample $\bm{{\rm X}}_1\sim p_1(\bm{{\rm X}}_1)$. In this work, we adopt the standard forward-time parameterization of the transport path. Specifically, the path is written as
\begin{equation}\label{path}
	\setlength\abovedisplayskip{3pt}
	\setlength\belowdisplayskip{3pt}
	\bm{{\rm X}}_t = (1-t)\bm{{\rm X}}_0 + t\bm{{\rm X}}_1, \quad t\in [0, 1],
\end{equation}
where $t=0$ corresponds to the prior state $\bm{{\rm X}}_0\sim p_0(\bm{{\rm X}}_0)$ and $t=1$ corresponds to the target state $\bm{{\rm X}}_1\sim q(\bm{{\rm X}}_1)$. For this conditional path, the corresponding velocity field $\bm{{\rm v}}(\bm{{\rm X}}_t, t|\bm{{\rm X}}_1)$ can be computed as the training target, which is given by
\begin{equation}
	\setlength\abovedisplayskip{3pt}
	\setlength\belowdisplayskip{3pt}
	\bm{{\rm v}}(\bm{{\rm X}}_t, t|\bm{{\rm X}}_1)=\bm{{\rm X}}_1 - \bm{{\rm X}}_0.
\end{equation}
Finally, the conditional instantaneous velocity prediction loss can be expressed as
\begin{equation}
	\setlength\abovedisplayskip{3pt}
	\setlength\belowdisplayskip{3pt}
	\mathcal{L}_{{\rm iv}}(\vartheta) = \mathbb{E}_{t, \bm{{\rm X}}_0, \bm{{\rm X}}_1}\left[\Vert \bm{{\rm v}}^\vartheta(\bm{{\rm X}}_t, t) - (\bm{{\rm X}}_1 - \bm{{\rm X}}_0)\Vert_F^2\right],
\end{equation}
where $t\sim \mathcal{U}[0,1]$, $\bm{{\rm X}}_0\sim p_0(\bm{{\rm X}}_0)$, and $\bm{{\rm X}}_1\sim q(\bm{{\rm X}}_1)$. Training $\bm{{\rm v}}^\vartheta(\bm{{\rm X}}_t, t)$ by minimizing this loss enables the model to learn the marginal instantaneous velocity field associated with the prescribed probability path. The corresponding generation process is iteratively executed according to
\begin{equation}\label{euler}
	\setlength\abovedisplayskip{3pt}
	\setlength\belowdisplayskip{3pt}
	\bm{{\rm \hat{X}}}_{t+\Delta_t} \gets \bm{{\rm \hat{X}}}_t + \Delta_t \bm{{\rm v}}^\vartheta(\bm{{\rm \hat{X}}}_t, t),
\end{equation}
where $t$ starts from $0$ and increases to $1$, $\bm{{\rm \hat{X}}}_0 = \bm{{\rm X}}_0$. We discretize the time interval $[0, 1]$ into $T$ uniform steps, with a step size of $\Delta_t = 1/T$. After $T$ forward updates, a generated approximation $\bm{{\rm \hat{X}}}_1$ is obtained.

\begin{figure}[t]
	\centering
	\includegraphics[width=3.5in]{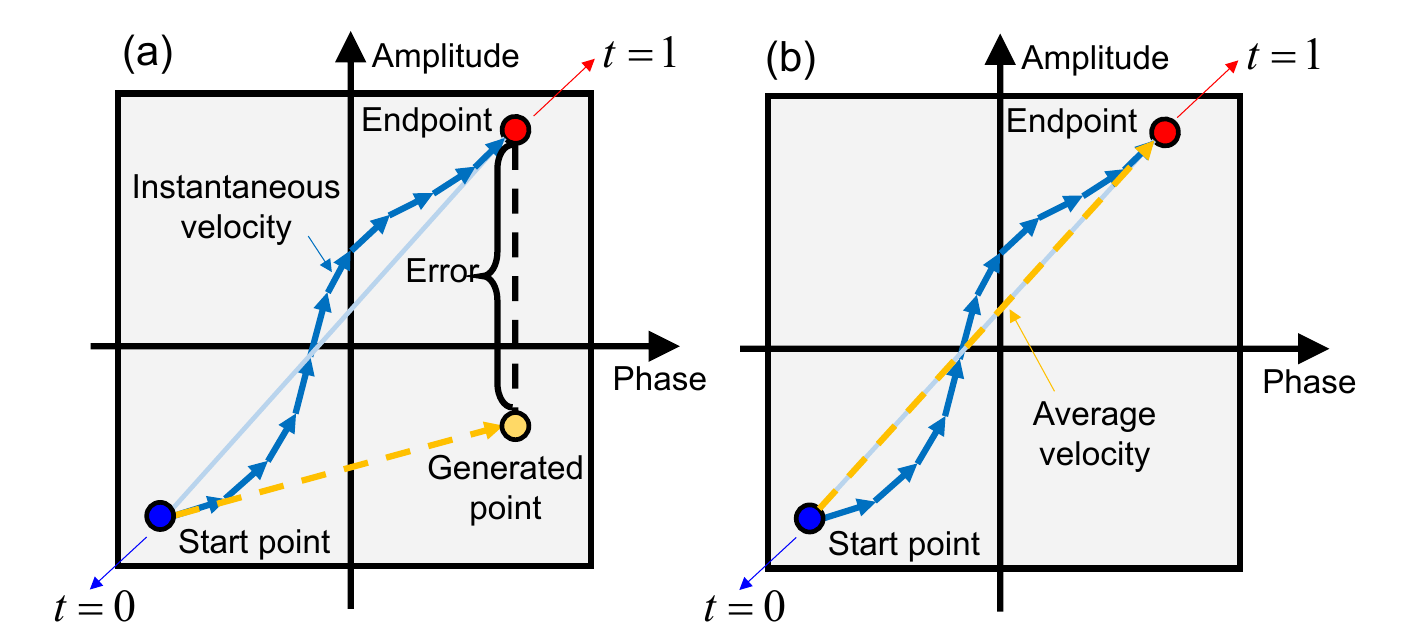}
	\caption{Comparisons of the discretization error in one-step generation, (a) learning instantaneous velocity and (b) learning average velocity.}
	\label{figure_1}
\end{figure}

\subsubsection{Learning Average Velocity Field}
For a model that predicts the instantaneous velocity field, numerical integration is required to evolve the state trajectory. When the number of integration steps $T$ is sufficiently large, solvers such as the Euler method in (\ref{euler}) can produce satisfactory states. However, when $T$ is small, the time discretization becomes coarse, making the numerical approximation less accurate and causing substantial error accumulation along the trajectory. As a result, the quality of the generated beamforming vectors may deteriorate noticeably. This quality-efficiency tradeoff becomes particularly pronounced in aggressive one-step beam generation, where one starts from the prior state $\bm{{\rm X}}_0$ at $t=0$ and applies only a single forward update over the whole interval:
\vspace*{-0.6\baselineskip}
\begin{equation}
	\setlength\abovedisplayskip{3pt}
	\setlength\belowdisplayskip{3pt}
	\bm{{\rm \hat{X}}}_1 \gets \bm{{\rm X}}_0 + \bm{{\rm v}}^\vartheta(\bm{{\rm X}}_0,0).
\end{equation}
This is because the instantaneous velocity evaluated only at the interval start cannot faithfully represent the overall transformation over the whole interval, which leads to significant discretization error as illustrated in Fig. \ref{figure_1}. To address this issue, we shift the learning object from instantaneous velocity to the average velocity over a time interval, whose forward-time formulation is directly compatible with explicit few-step generation. We therefore define the forward interval-average velocity field with respect to the interval-start state $\bm{{\rm X}}_r$ as
\begin{equation}\label{average_velocity}
	\setlength\abovedisplayskip{3pt}
	\setlength\belowdisplayskip{3pt}
	\begin{aligned}
		\bm{{\rm u}}(\bm{{\rm X}}_r, r, t) &\triangleq \frac{\bm{{\rm X}}_t-\bm{{\rm X}}_r}{t-r} \\
		&= \frac{1}{t-r}\int_r^t \bm{{\rm v}}(\bm{{\rm X}}_\tau, \tau)d\tau,\quad 0\le r < t \le 1.
	\end{aligned}
\end{equation}
By definition, the associated exact interval update is
\begin{equation}\label{exact_interval_update}
	\setlength\abovedisplayskip{3pt}
	\setlength\belowdisplayskip{3pt}
	\bm{{\rm X}}_t = \bm{{\rm X}}_r + (t-r)\bm{{\rm u}}(\bm{{\rm X}}_r, r, t).
\end{equation}
Eq. (\ref{average_velocity}) reveals the connection between the average and instantaneous velocities, and the latter can be viewed as the limit of the average velocity as the time interval length approaches
zero:
\begin{equation}
	\setlength\abovedisplayskip{3pt}
	\setlength\belowdisplayskip{3pt}
	\bm{{\rm u}}(\bm{{\rm X}}_t, t, t)\triangleq\lim_{r\to t}\bm{{\rm u}}(\bm{{\rm X}}_r, r, t) = \bm{{\rm v}}(\bm{{\rm X}}_t, t).
\end{equation}
However, directly learning (\ref{average_velocity}) is intractable, as it requires access to the full trajectory of the data sample. Inspired by \cite{guo2025splitmeanflow}, we introduce an intermediate time step $s$ such that $r < s < t$. By the additivity of definite integrals, we have
\begin{equation}\label{self-consistency}
	\setlength\abovedisplayskip{3pt}
	\setlength\belowdisplayskip{3pt}
	(t-r)\bm{{\rm u}}(\bm{{\rm X}}_r, r, t) \!=\! (s-r)\bm{{\rm u}}(\bm{{\rm X}}_r, r, s) + (t-s)\bm{{\rm u}}(\bm{{\rm X}}_s, s, t).
\end{equation}
Equation (\ref{self-consistency}) defines a self-consistent constraint that the average velocity field must satisfy across any partition of the time interval. Therefore, although the ground-truth average velocity is inaccessible, we can train a model $\bm{{\rm u}}^\vartheta(\bm{{\rm X}}_r, r, t)$ by enforcing this consistency during training. By rearranging (\ref{self-consistency}), we obtain
\begin{equation}
	\setlength\abovedisplayskip{3pt}
	\setlength\belowdisplayskip{3pt}
	\bm{{\rm u}}(\bm{{\rm X}}_r, r, t) = (1-\kappa)\bm{{\rm u}}(\bm{{\rm X}}_r, r, s)+\kappa\bm{{\rm u}}(\bm{{\rm X}}_s, s, t),
\end{equation}
where $\kappa = (t-s)/(t-r)\in [0,1]$. Hence, the split-consistency target becomes
\begin{equation}
	\setlength\abovedisplayskip{3pt}
	\setlength\belowdisplayskip{3pt}
	\bm{{\rm \bar{u}}}(\bm{{\rm X}}_r, r, t) = (1-\kappa)\bm{{\rm u}}^\vartheta(\bm{{\rm X}}_r, r, s)+\kappa\bm{{\rm u}}^\vartheta(\bm{{\rm \hat{X}}}_s, s, t),
\end{equation}
where $\bm{{\rm \hat{X}}}_s=\bm{{\rm X}}_r+(s-r)\bm{{\rm u}}^\vartheta(\bm{{\rm X}}_r, r, s)$. Accordingly, the loss for average-velocity learning is written as
\begin{equation}
	\setlength\abovedisplayskip{3pt}
	\setlength\belowdisplayskip{3pt}
	\mathcal{L}_{{\rm av}}(\vartheta) \!\!=\!\! \mathbb{E}_{r,s,t,\bm{{\rm X}}_0, \bm{{\rm X}}_1}\!\!\left[\!\Vert \bm{{\rm u}}^\vartheta(\bm{{\rm X}}_r,\! r, t) \!\!-\! {\rm sg}\left[\bm{{\rm \bar{u}}}(\bm{{\rm X}}_r, r, t)\right] \Vert_F^2\right]\!\!,
\end{equation}
where ${\rm sg}(\cdot)$ denotes the stop-gradient operator. The architectural implementation of prompt conditioning is described in the next subsection.

\begin{remark}[Exactness of the Interval-Average Velocity Representation]
	We now briefly elucidate why it can reproduce the original transport process. Recall that the original transport from $\bm{{\rm X}}_r$ to $\bm{{\rm X}}_t$ is governed by the instantaneous velocity ODE and therefore satisfies $\bm{{\rm X}}_t=\bm{{\rm X}}_r+\int_r^t \bm{{\rm v}}(\bm{{\rm X}}_{\tau}, \tau)d\tau$. By definition, the integral displacement is exactly equal to $(t-r)\bm{{\rm u}}(\bm{{\rm X}}_r, r, t)$, yielding (\ref{exact_interval_update}). Hence, the update using average velocity is an equivalent representation of the original ODE transport rather than a numerical approximation. In contrast, the instantaneous velocity update approximates the integral $\int_r^t \bm{{\rm v}}(\bm{{\rm X}}_{\tau}, \tau)d\tau$ using $(t-r)\bm{{\rm v}}(\bm{{\rm X}}_r, r)$. Thus, a time-discretization error is introduced whenever the velocity varies along the trajectory.
\end{remark}
\vspace*{-0.6\baselineskip}
While the exact average velocity $\bm{{\rm u}}(\bm{{\rm X}}_r, r, t; \bm{{\rm c}}^\star)$ eliminates time discretization error, the accuracy of the learned field $\bm{{\rm u}}^\vartheta(\bm{{\rm X}}_r, r, t; \bm{{\rm c}})$ is affected by both model approximation and practical probing imperfections, where $\bm{{\rm c}}^\star$ and $\bm{{\rm c}}$ denote the ideal and actual wireless prompts, respectively. In practice, their mismatch may arise from measurement noise or feedback delay. We next characterize how these errors affect the transport and propagate during few-step inference.
\vspace*{-0.2\baselineskip}

\begin{proposition}[Terminal Error Bound for Learned Average-Velocity Transport under Prompt Perturbations]\label{propositon1}
	Consider a partition $0=t_0<\cdots<t_T=1$, where $\Delta_i=t_{i+1}-t_i$, $i=0,\ldots,T-1$. Let $\bm{{\rm c}}^\star$ and $\bm{{\rm c}}$ denote the ideal and perturbed wireless prompts, respectively, and define
	\begin{equation}
		\setlength\abovedisplayskip{3pt}
		\setlength\belowdisplayskip{3pt}
		\delta_c = \Vert \bm{{\rm c}}-\bm{{\rm c}}^\star \Vert_F.
	\end{equation}
	Suppose that, over a domain $\mathcal{X}$ containing both the ideal trajectory $\{\bm{{\rm X}}_{t_i}\}_{i=0}^T$ and the generated trajectories $\{\bm{\hat{{\rm X}}}_{t_i}\}_{i=0}^T$, the learned average velocity field satisfies
	\begin{equation}
		\setlength\abovedisplayskip{3pt}
		\setlength\belowdisplayskip{3pt}
		\Vert \bm{{\rm u}}^\vartheta(\bm{{\rm X}}, t_i, t_{i+1}; \bm{{\rm c}}^\star)-\bm{{\rm u}}(\bm{{\rm X}}, t_i, t_{i+1}; \bm{{\rm c}}^\star)\Vert_F \leq \epsilon_i,
	\end{equation}
	where $\epsilon_i$ is defined as the worst-case prediction error over the $i$-th time interval. Further, assume that $\bm{{\rm u}}^{\vartheta}$ is $L_u$- and $L_c$-Lipschitz continuous with respect to its state and condition, respectively. Therefore, we have
	\begin{equation}\label{eq:lipschitz_state}
		\setlength\abovedisplayskip{3pt}
		\setlength\belowdisplayskip{3pt}
		\left\|\bm{{\rm u}}^\vartheta(\bm{{\rm X}}, t_i, t_{i+1}; \bm{{\rm c}})
		-\bm{{\rm u}}^\vartheta(\bm{{\rm Y}}, t_i, t_{i+1}; \bm{{\rm c}})\right\|_F
		\leq L_u\left\|\bm{{\rm X}}-\bm{{\rm Y}}\right\|_F,
	\end{equation}
	for all $\bm{{\rm X}}, \bm{{\rm Y}}\in \mathcal{x}$ and relevant prompts $\bm{{\rm c}}$, and 
	\begin{equation}\label{eq:lipschitz_condition}
		\setlength\abovedisplayskip{3pt}
		\setlength\belowdisplayskip{3pt}
		\left\|\bm{{\rm u}}^\vartheta(\bm{{\rm X}}, t_i, t_{i+1}; \bm{{\rm c}}_1)
		\!-\!\bm{{\rm u}}^\vartheta(\bm{{\rm X}}, t_i, t_{i+1}; \bm{{\rm c}}_2)\right\|_F
		\leq L_c\left\|\bm{{\rm c}}_1-\bm{{\rm c}}_2\right\|_F,
	\end{equation}
	for all $\bm{{\rm X}} \in \mathcal{X}$ and relevant prompts $\bm{{\rm c}}_1$ and $\bm{{\rm c}}_2$.
	If the ideal and generated trajectories start from the same initial state, then the terminal generation error satisfies
	\begin{equation}\label{eq:bound}
		\setlength\abovedisplayskip{3pt}
		\setlength\belowdisplayskip{3pt}
		\Vert \bm{{\rm \hat{X}}}_{1}-\bm{{\rm X}}_{1}\Vert_F
		\leq 
		\sum_{i=0}^{T-1}\Delta_i (\epsilon_i + L_c\delta_c)\prod_{j={i+1}}^{T-1}(1+L_u\Delta_j)
	\end{equation}
	In particular, one-step generation satisfies $\Vert \bm{{\rm \hat{X}}}_1\!\!-\!\!\bm{{\rm X}}_{1}\Vert_F \!\leq \epsilon_0 + L_c\delta_c$.
\end{proposition}

\begin{IEEEproof}
	Define the state discrepancy at $t_i$ as $e_i\triangleq \|\bm{{\rm \hat{X}}}_{t_i}-\bm{{\rm X}}_{t_i}\|_F$, $i=0,\cdots,T$. For the $i$-th interval $[t_i, t_{i+1}]$, the ideal conditional transport satisfies
	\begin{equation}
		\setlength\abovedisplayskip{3pt}
		\setlength\belowdisplayskip{3pt}
		\bm{{\rm X}}_{t_{i+1}} = \bm{{\rm X}}_{t_i} + \Delta_{i} \bm{{\rm u}}(\bm{{\rm X}}_{t_i}, t_i, t_{i+1}; \bm{{\rm c}}^\star),
	\end{equation}
	whereas the learned transport under the perturbed prompt is
	\begin{equation}
		\setlength\abovedisplayskip{3pt}
		\setlength\belowdisplayskip{3pt}
		\bm{{\rm \hat{X}}}_{t_{i+1}} = \bm{{\rm \hat{X}}}_{t_i} + \Delta_{i} \bm{{\rm u}}^{\vartheta}(\bm{{\rm \hat{X}}}_{t_i}, t_i, t_{i+1}; \bm{{\rm c}}).
	\end{equation}
	Subtracting the ideal update from the learned update and applying the triangle inequality gives
	\begin{equation}\label{eq:error_triangle_inequality}
		\setlength\abovedisplayskip{3pt}
		\setlength\belowdisplayskip{3pt}
		e_{i+1} \! \leq \! e_i \!+\! \Delta_{i} \Vert \bm{{\rm u}}^{\vartheta}(\bm{{\rm \hat{X}}}_{t_i}, t_i, t_{i+1}; \bm{{\rm c}}) \!-\! \bm{{\rm u}}(\bm{{\rm X}}_{t_i}, t_i, t_{i+1}; \bm{{\rm c}}^\star) \Vert_F.
	\end{equation}
	The velocity discrepancy can be decomposed as
	\begin{align}
		& \Vert \bm{{\rm u}}^{\vartheta}(\bm{{\rm \hat{X}}}_{t_i}, t_i, t_{i+1}; \bm{{\rm c}}) - \bm{{\rm u}}(\bm{{\rm X}}_{t_i}, t_i, t_{i+1}; \bm{{\rm c}}^\star) \Vert_F \nonumber \\
		& \leq \Vert \bm{{\rm u}}^{\vartheta}(\bm{{\rm \hat{X}}}_{t_i}, t_i, t_{i+1}; \bm{{\rm c}}) - \bm{{\rm u}}^{\vartheta}(\bm{{\rm X}}_{t_i}, t_i, t_{i+1}; \bm{{\rm c}}) \Vert_F \nonumber \\
		&+ \Vert \bm{{\rm u}}^{\vartheta}(\bm{{\rm X}}_{t_i}, t_i, t_{i+1}; \bm{{\rm c}}) - \bm{{\rm u}}^{\vartheta}(\bm{{\rm X}}_{t_i}, t_i, t_{i+1}; \bm{{\rm c}}^\star) \Vert_F \nonumber \\
		&+\Vert \bm{{\rm u}}^{\vartheta}(\bm{{\rm X}}_{t_i}, t_i, t_{i+1}; \bm{{\rm c}}^\star) - \bm{{\rm u}}(\bm{{\rm X}}_{t_i}, t_i, t_{i+1}; \bm{{\rm c}}^\star) \Vert_F, \label{eq:prediction_error}
	\end{align}
	where the first term characterizes the velocity deviation caused by the \textit{accumulated state error}, the second term represents the velocity deviation caused by the \textit{perturbed wireless prompt}, and the third term represents the \textit{velocity approximation error} inherent in the model itself. Based on (\ref{eq:lipschitz_state}), (\ref{eq:lipschitz_condition}), we have
	\begin{equation}\label{eq:inequality_joint_state_and_condition}
		\setlength\abovedisplayskip{3pt}
		\setlength\belowdisplayskip{3pt}
		\Vert \bm{{\rm u}}^{\vartheta}(\bm{{\rm \hat{X}}}_{t_i}, t_i, t_{i+1}; \bm{{\rm c}}) - \bm{{\rm u}}(\bm{{\rm X}}_{t_i}, t_i, t_{i+1}; \bm{{\rm c}}^\star) \Vert_F 
		\leq L_ue_i + L_c\delta_c + \epsilon_i.
	\end{equation}
	Substituting (\ref{eq:inequality_joint_state_and_condition}) into (\ref{eq:error_triangle_inequality}), we obtain
	\begin{equation}\label{eq:recurrence}
		\setlength\abovedisplayskip{3pt}
		\setlength\belowdisplayskip{3pt}
		e_{i+1} \leq (1+L_u \Delta_{i}) e_i + \Delta_{i} (\epsilon_i + L_c\delta_c).
	\end{equation}
	Since the ideal and generated trajectories share the same initial state, i.e., $e_0 = \Vert \bm{{\rm \hat{X}}}_0 - \bm{{\rm X}}_0 \Vert_F=0$, recursively applying the error inequality yields
	\begin{equation}
		\setlength\abovedisplayskip{3pt}
		\setlength\belowdisplayskip{3pt}
		e_T \leq \sum_{i=0}^{T-1}\Delta_i (\epsilon_i + L_c\delta_c)\prod_{j={i+1}}^{T-1}(1+L_u\Delta_j).
	\end{equation}
	where the empty product corresponding to $i=T-1$ is defined as one. Since $t_T=1$, we have $e_T = \Vert \bm{{\rm \hat{X}}}_1 - \bm{{\rm X}}_1 \Vert_F$. Therefore, 
	\begin{equation}
		\setlength\abovedisplayskip{3pt}
		\setlength\belowdisplayskip{3pt}
		\Vert \bm{{\rm \hat{X}}}_1 - \bm{{\rm X}}_1 \Vert_F \leq \sum_{i=0}^{T-1}\Delta_i (\epsilon_i + L_c\delta_c)\prod_{j={i+1}}^{T-1}(1+L_u\Delta_j).
	\end{equation}
	When $T=1$, we have $\Vert \bm{{\rm \hat{X}}}_1\!\!-\!\!\bm{{\rm X}}_{1}\Vert_F \!\leq \epsilon_0 + L_c\delta_c$. This completes the proof.
\end{IEEEproof}

\textbf{Proposition} \textbf{\ref{propositon1}} decomposes the terminal transport error into three factors: the average-velocity approximation errors $\{\epsilon_i\}_{i=0}^{T-1}$, the practical probing error $L_c\delta_c$, and the multi-step error amplification governed by $L_u$. Building upon this, we conduct further analysis and uncover several insights for GenSSBF:
\begin{itemize}
	\item The error bound in (\ref{eq:bound}) characterizes a trade-off between prediction accuracy and error propagation. In practice, increasing $T$ brings shorter intervals, which may reduce the corresponding
	approximation errors $\epsilon_i$, but introduces additional
	inter-step error propagation. Therefore, one-step generation is preferable when the full-interval average velocity can be predicted accurately, whereas few-step generation can be beneficial when the reduction in interval-wise approximation error outweighs the additional propagation.
	
	\item The terminal error explicitly separates the average velocity approximation error from the practical probing error. This decomposition identifies two complementary design directions: improving average velocity learning to reduce $\epsilon_i$, and enhancing prompt robustness by reducing the condition mismatch $\delta_c$ or its sensitivity $L_c$. As a direct consistency result, $\max_i\epsilon_i\rightarrow0$ and $\delta_c\rightarrow0$ imply $e_T\rightarrow0$ for any finite partition.
\end{itemize}

\vspace{-0.5\baselineskip}
\subsection{Model Architecture}\label{section3.B}
The architecture of the velocity predictor in the proposed FBBS framework is illustrated in Fig. \ref{figure_2}. To realize the beam evolution strategy introduced in Section \ref{section3.A}, we develop a customized one-dimensional diffusion transformer (DiT)-style backbone that operates on beam representations and supports both instantaneous- and interval-conditioned velocity prediction. Given the current state $\bm{{\rm X}}_r\in \mathbb{R}^{B\times C\times N_t}$, the time interval $(r,t)$, and the condition $\bm{{\rm c}}_k$, the model outputs a velocity tensor with the same dimensionality as $\bm{{\rm X}}_r$, where $B$ is the batch size, $C$ denotes the number of input channels (e.g., real and imaginary components), and $N_t$ is the number of transmit antenna elements. This backbone is shared by both learning stages of FBBS: in Stage I, it operates in instantaneous mode and predicts $\bm{{\rm u}}^\vartheta(\bm{{\rm X}}_r,r,r,\bm{{\rm c}}_k)$; in Stage II, it operates in interval mode and predicts $\bm{{\rm u}}^\vartheta(\bm{{\rm X}}_r,r,t,\bm{{\rm c}}_k)$.

\begin{figure}[t]
	\centering
	\includegraphics[width=3.0in]{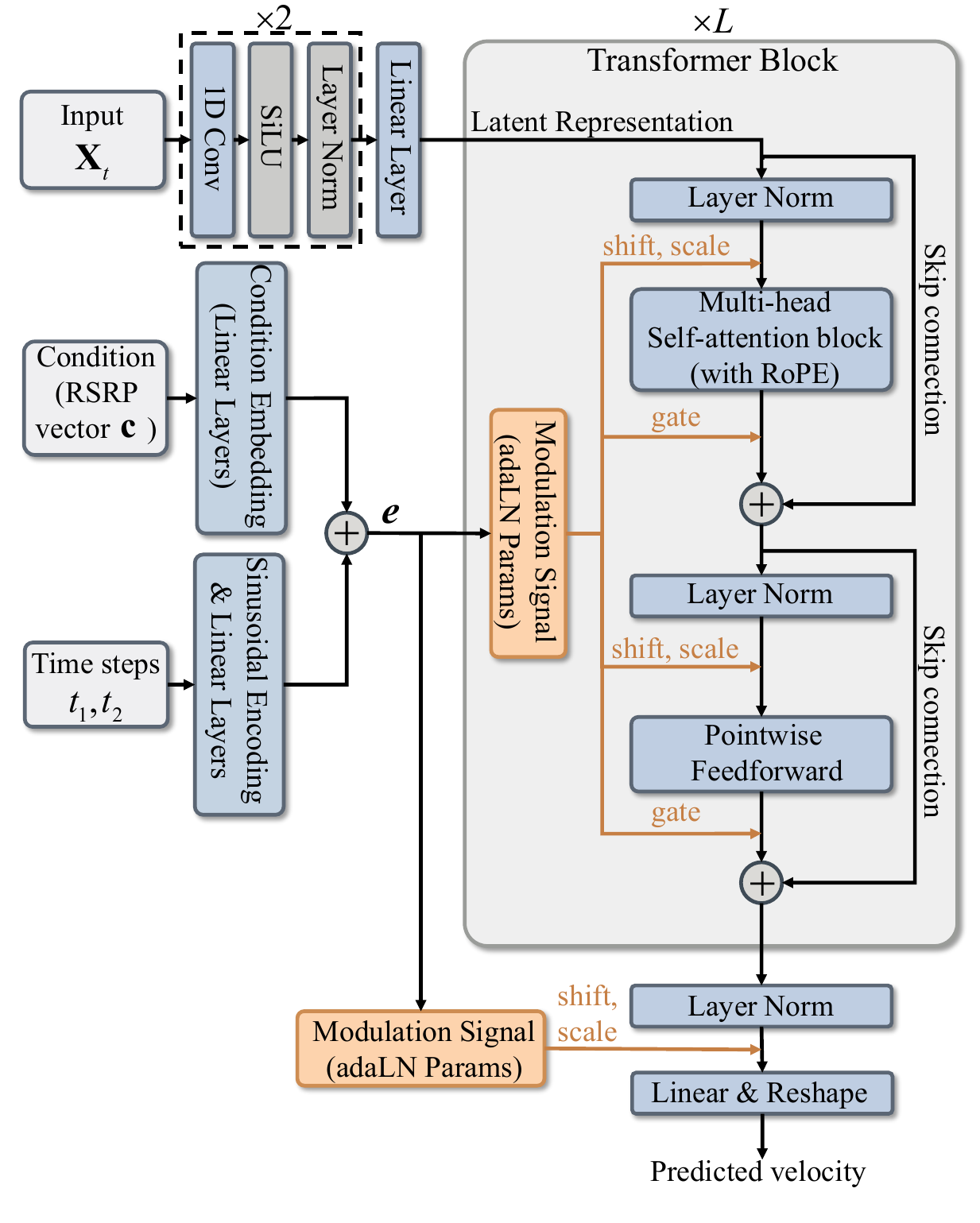}
	\caption{The architecture of the 1D-DiT.}
	\label{figure_2}
\end{figure}

The backbone consists of a stack of $L$ transformer blocks, each containing a multi-head self-attention module and a feed-forward network (FFN). This design is adopted because beam evolution is inherently a structured sequence modeling problem: the desired update at one beam component depends not only on its local value, but also on its interaction with other components across the entire beam vector. To capture such dependencies, rotary position embedding (RoPE) is incorporated into the self-attention module, enabling the network to encode relative positional relationships along the one-dimensional beam sequence while preserving sensitivity to global structural patterns.

Conditioning is injected through adaptive layer normalization (adaLN)\cite{peebles2023scalable}, as shown by the orange line in Fig. \ref{figure_2}, where the normalization parameters are modulated by the total embedding vector $\bm{{\rm e}}_{\text{total}}=\bm{{\rm e}}_{\text{time}}(r,t)+\bm{{\rm e}}_{\text{cond}}$. Here, $\bm{{\rm e}}_{\text{time}}(r,t)$ denotes the embedding of the time interval $(r,t)$, while $\bm{{\rm e}}_{\text{cond}}$ denotes the condition embedding defined in Section \ref{section3.C}. This vector is then linearly projected to six modulation parameters, namely $\beta_1$, $\gamma_1$, and $\alpha_1$ for the attention branch, and $\beta_2$, $\gamma_2$, and $\alpha_2$ for the FFN branch. These parameters control feature-wise shifting, scaling, and gating inside each transformer block, allowing the wireless prompt and time interval to modulate feature extraction throughout the entire network. After passing through the $L$ transformer blocks, a final adaLN layer and a linear projection are applied to map the latent representation back to the original channel dimension $C$, yielding the predicted velocity tensor.
\vspace*{-0.8\baselineskip}
\subsection{Flexible Probing via Stochastic Masking}\label{section3.C}
Another key limitation of the existing GenSSBF methods is that they typically assume a fixed-length condition, requiring separate models when the number of available RSRP measurements changes. To enable flexible probing with a single model, we apply stochastic masking to the RSRP condition during training. Let $\bm{{\rm o}}_{k,Q_{{\rm max}}}\in \mathbb{R}^{Q_{{\rm max}}}$ denote the complete RSRP measurements and $\bm{{\rm m}}\in \{0, 1\}^{Q_{{\rm max}}}$ its binary availability mask. For each mini-batch of size $B$, a \emph{sparsity controller} $p_{\text{full}}\in(0,1]$ specifies the fraction of samples that retain the full conditioning information. This controller can operate under either a constant or a dynamic sparsity schedule. In particular, the former maintains a fixed sparsity level throughout training, while the latter sets sparsity $p_{\text{full}}$ as an explicit function of the training epoch. Thus, $B_{\rm full}=\lfloor p_{\rm full}B\rfloor$ samples retain the complete condition ($\mathbf{m}=\mathbf{1}$), while for the remaining $B-B_{\text{full}}$ samples in this batch, we randomly assign probing budgets and apply corresponding masks. This design is motivated by the following considerations: First, full-condition samples provide a high-information supervision anchor, preventing the model from drifting toward solutions tailored only to overly sparse conditions. Second, they stabilize optimization in both Stage I and II in Section \ref{section3.A} by reducing gradient variance induced by aggressive masking and by improving the reliability of split-interval training targets. Third, they preserve calibration to the high-information regime, ensuring that the same model remains competitive when larger probing budgets are available at inference. The active probing indices are selected by uniform spacing over the codebook grid, which preserves global angular coverage under reduced feedback. The masked condition is thus given by $\bm{{\rm c}}_k = \bm{{\rm m}}\odot \bm{{\rm o}}_{k,Q_{{\rm max}}}$, and is standardized prior to input. It should be noted that the masked condition is generated on-the-fly and used in both training stages, avoiding dataset expansion and budget-specific retraining. Since zero-padding the RSRP vector alone does not identify which probing beams are observed, we further employ a beam index-aware condition encoder. Specifically, the condition embedding is computed as
\vspace{-0.1\baselineskip}
\begin{equation}
	\setlength\abovedisplayskip{3pt}
	\setlength\belowdisplayskip{3pt}
	\bm{{\rm e}}_{\rm cond}
	=
	f_{\rm out}\!\left(
	\frac{1}{\|\bm{{\rm m}}\|_1}
	\sum_{i=1}^{Q_{\max}}
	m_i
	f_{\rm pos}\!\left(
	\bm{{\rm W}}_v\bm{{\rm c}}_k[i]+\bm{{\rm b}}_v+\bm{{\rm p}}_i
	\right)
	\right),
\end{equation}
where $\bm{{\rm W}}_v$ and $\bm{{\rm b}}_v$ are the weights and bias of the point-wise linear layer. $\bm{{\rm p}}_i$ is the sinusoidal embedding of the $i$-th probing beam index, and $f_{\rm pos}(\cdot)$ and $f_{\rm out}(\cdot)$ are position-wise nonlinear projectors. The index embedding preserves the semantics of the observed beams, while normalization by $\|\bm{{\rm m}}\|_1$ prevents the embedding scale from varying with the probing budget. The resulting fixed-dimensional $\bm{{\rm e}}_{\rm cond}$ is combined with the time embedding through the adaLN modulation described in Section III-B. Consequently, a single FBBS model can accommodate different probing budgets without retraining. At deployment, the base station can select any $Q_{\text{infer}}\in \mathcal{Q}_p$ without retraining. A small $Q_{\text{infer}}$ is suitable for low-latency operation with minimal sweeping overhead, whereas a larger $Q_{\text{infer}}$ provides richer conditioning when the system can afford more information. Probing budgets can also be selected adaptively according to mobility, channel dynamics, or quality-of-service (QoS) requirements.

\begin{remark}[Lipschitz Continuity of the Adopted Backbone]
	For a fixed availability mask $\bm{{\rm m}}$, the proposed 1D-DiT is composed of finite-dimensional affine mappings, SiLU activations, softmax attention, RoPE transformations, and normalization layers with strictly positive numerical constants. Moreover, the time variables $(r,t)$ and the standardized condition vector lie in bounded admissible domains. Therefore, for finite learned parameters, the predictor has bounded Jacobians on the considered compact domain. Consequently, finite constants $L_u$ and $L_c$ exist such that the Lipschitz conditions in (\ref{eq:lipschitz_state}) and (\ref{eq:lipschitz_condition}) hold.
\end{remark}

\subsection{Training and Deployment}\label{section3.D}
\begin{algorithm}[t]
	\caption{Offline Training of FBBS}
	\label{algorithm1}
	\begin{algorithmic}[1]
	\STATE \textbf{Input:} dataset $\mathcal{D}=\{(\bm{{\rm X}}_1^d,\bm{{\rm c}}_d)\}_{d=1}^D$, total epochs $E$, stage-I epochs $E_1$, boundary ratio $p$, full-condition ratio $p_{\rm full}$, probing budget set $\mathcal{Q}_{p}$, batch size $B$
	\STATE \textbf{Output:} trained parameters $\vartheta$
	
	\FOR{$e=1$ to $E$}
	\FOR{mini-batch $(\mathbf X_1,\mathbf c)\sim\mathcal D$}
	\STATE Sample $\mathbf X_0\sim\mathcal N(\mathbf 0,\mathbf I)$
	\STATE $(\widetilde{\mathbf c},\mathbf m)\gets\operatorname{Mask}(\mathbf c;\mathcal Q_p,p_{\rm full})$
	\IF{$e\leq E_1$}
		\STATE Sample $t\sim\mathcal U[0,1]^B$
		\STATE $\mathbf X_t\gets(1-t)\mathbf X_0+t\mathbf X_1$
		\STATE Compute the instantaneous velocity loss $\mathcal L_{\rm iv}=\|u_{\vartheta}(\mathbf X_t,t,t;\widetilde{\mathbf c},\mathbf m)-(\mathbf X_1-\mathbf X_0)\|_F^2$
	\ELSE
		\STATE Sample and sort two time points $0\leq r\leq t\leq1$
		\STATE Randomly set $r=t$ for $\lfloor pB\rfloor$ samples
		\STATE Sample $\kappa\sim\mathcal U[0,1]^B$ and set $s=(1-\kappa)t+\kappa r$
		\STATE $\mathbf X_r\gets(1-r)\mathbf X_0+r\mathbf X_1$
		\STATE $\widehat{\mathbf u}_{r,t}\gets {\mathbf u}_{\vartheta}(\mathbf X_r,r,t; \widetilde{\mathbf c},\mathbf m)$
		\STATE $\widehat{\mathbf u}_{r,s}\gets {\mathbf u}_{\vartheta}(\mathbf X_r,r,s;\widetilde{\mathbf c},\mathbf m)$
		\STATE $\widehat{\mathbf X}_s\gets\mathbf X_r+(s-r)\widehat{\mathbf u}_{r,s}$
		\STATE $\widehat{\mathbf u}_{s,t}\gets {\mathbf u}_{\vartheta}(\widehat{\mathbf X}_s,s,t;\widetilde{\mathbf c},\mathbf m)$
		\STATE Form target $\mathbf u_{\rm tgt}\gets (1-\kappa)\widehat{\mathbf u}_{r,s}+\kappa\widehat{\mathbf u}_{s,t}$
		\STATE For boundary samples with $t=r$, replace the target as $\mathbf u_{\rm tgt}=\mathbf X_1-\mathbf X_0$
		\STATE Compute the average velocity loss $\mathcal L_{\rm av}=\|\widehat{\mathbf u}_{r,t}-\operatorname{sg}[\mathbf u_{\rm tgt}]\|_F^2$
	\ENDIF
	\STATE Update $\vartheta$ using $\mathcal L_{\rm iv}$ if $e\leq E_1$, and $\mathcal L_{\rm av}$ otherwise
	\ENDFOR
	\ENDFOR
		
	\end{algorithmic}
\end{algorithm}		


The offline training and online inference procedures of FBBS are summarized in \textbf{Algorithm} \textbf{\ref{algorithm1}} and \textbf{Algorithm} \textbf{\ref{algorithm2}}, respectively. During training, the stochastic masking strategy in Section III-C is applied on-the-fly to each mini-batch. Stage I minimizes $\mathcal{L}_{\rm iv}$ over the first $E_1$ epochs to learn the instantaneous velocity field, while Stage II continues from the Stage-I model and minimizes $\mathcal{L}_{\rm av}$ over the remaining epochs to learn the interval-conditioned average velocity. Applying the same masking strategy in both stages enables a single model to support variable probing budgets.
	
At deployment, the BS probes $Q$ DFT beams and constructs the masked prompt $(\bm{{\rm c}}_k,\bm{{\rm m}}_k)$. Conditioned on this prompt, FBBS evolves $M$ independent prior states through $T$ average velocity-guided updates. Each terminal state $\bm{{\rm \hat{X}}}_1^{k,m}$ is mapped to a constant-modulus beamforming candidate according to
\begin{equation}\label{eq:beam_recovery}
\left\{
\begin{aligned}
	&\bm{{\rm h}}^{\mathcal{A}}_{k,m} = \bm{{\rm \hat{X}}}_1^{k,m}[1, :] + j\bm{{\rm \hat{X}}}_1^{k,m}[2, :],\\
	&\bm{{\rm w}}_{k,m} = \frac{1}{\sqrt{N_t}}{\rm e}^{j\angle \mathcal{F}^{-1}(\bm{{\rm h}}^{\mathcal{A}}_{k,m})}.
\end{aligned}
\right.
\end{equation}
The $M$ candidates are then probed, and the one yielding the highest RSRP is selected. Therefore, the brainstorming number $M$ serves as a practical control knob between online overhead and beamforming robustness: a small $M$ is suitable for latency-sensitive access, whereas a moderately larger $M$ is beneficial when the system can afford slightly more verification cost to pursue higher beamforming gain. \textbf{Proposition} 2 explains the effectiveness of beam brainstorming in improving beamforming gain.

\begin{algorithm}[t]
	\caption{Online Inference of FBBS}
	\label{algorithm2}
	\begin{algorithmic}[1]
		\STATE \textbf{Input:} probing budget $Q$, generation steps $T$, brainstorm number $M$, and trained predictor $\bm{{\rm u}}^\vartheta$
		\STATE \textbf{Output:} Selected beamforming vector $\widehat{\bm w}_k$
		
		\STATE Probe $Q$ uniformly spaced DFT beams and form the prompt $(\bm{{\rm c}}_k,\bm{{\rm m}}_k)$
		\STATE Draw $\bm{{\rm X}}_{0}^{k,m}\overset{\mathrm{i.i.d.}}{\sim}\mathcal N(\bm 0,\bm I)$ for all $m\in\{1,\ldots,M\}$ and set $\Delta=1/T$, $\bm{{\rm \hat{X}}}_0^{k,m}\leftarrow \bm{{\rm X}}_{0}^{k,m}$
		
		\FOR{$t=0$ to $T-1$}
		\STATE $\bm{{\rm \hat{X}}}_{(t+1)\Delta}^{k,m}\!\leftarrow\!\bm{{\rm \hat{X}}}_{t\Delta}^{k,m}\!\!+\!\!{\mathbf u}^{\vartheta}\!\left(\!\bm {{\rm \hat{X}}}_{t\Delta}^{k,m},\!t\Delta,(t\!+\!1)\Delta; \widetilde{\bm {{\rm c}}}_k,\bm{{\rm m}}_k\right)\!\!\Delta$
		\ENDFOR
		
		\STATE Recover $\{\bm{{\rm w}}_{k,m}\}_{m=1}^{M}$ from $\{\bm{{\rm \hat{X}}}_{1}^{k,m}\}_{m=1}^{M}$ using (\ref{eq:beam_recovery})
		
		\STATE Probe $\{\bm{{\rm w}}_{k,m}\}_{m=1}^{M}$
		\STATE Set $m^\star\leftarrow\arg\max_m\operatorname{RSRP}(\bm{{\rm w}}_{k,m})$
		\STATE \textbf{return} $\widehat{\bm{{\rm w}}}_k\leftarrow\bm{{\rm w}}_{k,m^\star}$
	\end{algorithmic}
\end{algorithm}

\begin{proposition}[Probabilistic Coverage of Beam Brainstorming]
	\label{prop:brainstorm_coverage}
	For UE $k$ and its wireless prompt, let $\bar{g}_{k,m}$ denote the normalized beamforming gain of candidate $m$. Conditioned on the wireless prompt, suppose that $\{\bar{g}_{k,m}\}_{m=1}^M$ are independently and identically distributed. For a target gain $\gamma$, define the single-candidate success probability as
	$p_{\gamma,k}\triangleq\Pr(\bar{g}_{k,m}\geq\gamma)$. Assuming reliable candidate selection, the probability that the selected beam achieves the target gain is
	\begin{equation}\label{eq:brainstorm_coverage}
		\setlength\abovedisplayskip{3pt}
		\setlength\belowdisplayskip{3pt}
		P_{\gamma,k}^{(M)}=1-(1-p_{\gamma,k})^M.
	\end{equation}
	Moreover, the marginal gain from adding one additional candidate beam is
	\begin{equation}
		\setlength\abovedisplayskip{3pt}
		\setlength\belowdisplayskip{3pt}
		P_{\gamma,k}^{(M+1)}-P_{\gamma,k}^{(M)}=p_{\gamma,k}(1-p_{\gamma,k})^M,
	\end{equation}
	which is non-negative and decreases with $M$.
\end{proposition}

\begin{IEEEproof}
	The selected beam fails to reach $\gamma$ only when all $M$ independently generated candidates fail, whose probability is $(1-p_{\gamma,k})^M$. The marginal improvement follows directly from (\ref{eq:brainstorm_coverage}).
\end{IEEEproof}

\textbf{Proposition}~\ref{prop:brainstorm_coverage} indicates that increasing $M$ improves gain coverage but yields diminishing returns. Hence, to guarantee $P_{\gamma,k}^{(M)}\geq1-\eta$, the minimum required brainstorming number is $M_{\min}=\left\lceil\frac{\log\eta}{\log(1-p_{\gamma,k})}\right\rceil$ when $0<p_{\gamma, k}<1$ and $0<\eta<1$.

\section{Numerical Results}\label{Section4}
\subsection{Scenarios}\label{Section4.A}
To comprehensively evaluate the proposed FBBS under diverse propagation conditions, three ray-tracing datasets are considered in this work. Specifically, we first construct HKU\_28, a campus dataset of The University of Hong Kong (HKU) generated using SionnaRT\cite{hoydis2023sionna}, and adopt two public scenarios, namely O1B\_28 and I2\_28B, from the DeepMIMO dataset\cite{alkhateeb2019deep}. These three datasets cover both indoor and outdoor environments, as well as LoS and NLoS propagation conditions, thereby providing a comprehensive testbed for assessing the robustness and practical effectiveness of the proposed FBBS. 

\subsubsection{HKU\_28 Scenario}
As shown in Fig. \ref{fig:scenarios}(a), we created the HKU\_28 dataset, which models the campus environment of HKU and features a mix of high-rise buildings, streets, and open spaces. Specifically, the 3D model of the HKU campus was constructed by importing OpenStreetMap (OSM) data into Blender. Subsequently, ray-tracing simulations were performed on this scene using Sionna RT to generate the dataset. The BS is located at a height of 15 meters, and a total of 82169 users are activated in this scenario. The carrier frequency is 28 GHz.

\subsubsection{DeepMIMO O1B\_28 Scenario}
The DeepMIMO O1B\_28 dataset models an outdoor street environment with blockage and reflections, as shown in Fig. \ref{fig:scenarios}(b). Specifically, a portion of the O1\_28 dataset corresponding to BS \#3 and user grid \#1 is selected. Therefore, the total population of 484264 users consists of those with LoS links to the BS and those with NLoS links. The carrier frequency is 28 GHz.

\subsubsection{DeepMIMO I2\_28B Scenario}
As illustrated in Fig. \ref{fig:scenarios}(c), the DeepMIMO I2\_28B dataset models an indoor scenario where none of the users has a LoS connection with the BS. The BS is placed on the inside wall of the room, and a total of 140901 users are distributed in the brown-colored area located behind the blockage. The carrier frequency is 28 GHz.

In this work, static ray-tracing datasets are used to evaluate the performance of FBBS. In practice, site-specific training pairs can be obtained either from calibrated digital twins or through periodic site surveys and channel sounding campaigns. The model is intended to capture slowly varying propagation structures rather than individual transient channel realizations. Therefore, short-term changes caused by UE mobility or temporary blockage can generally be handled through updated wireless prompts, and their effects are partially examined through noisy- and outdated prompt evaluations in Sections~\ref{sec:noisy_prompts} and \ref{sec:outdated_prompts}. In contrast, persistent changes that alter the underlying channel distribution may require fine-tuning or retraining. The cross-site generalizability of FBBS is further evaluated in Table~\ref{tab:cross_site_generalization}. Therefore, no universal update interval is assumed. Model updates may instead be triggered by detected long-term site changes or sustained performance degradation. The corresponding maintenance cost is site-specific and mainly arises from updating the site data and performing occasional offline model adaptation.

\begin{figure}[t]
	\centering
	\includegraphics[width=3.3in]{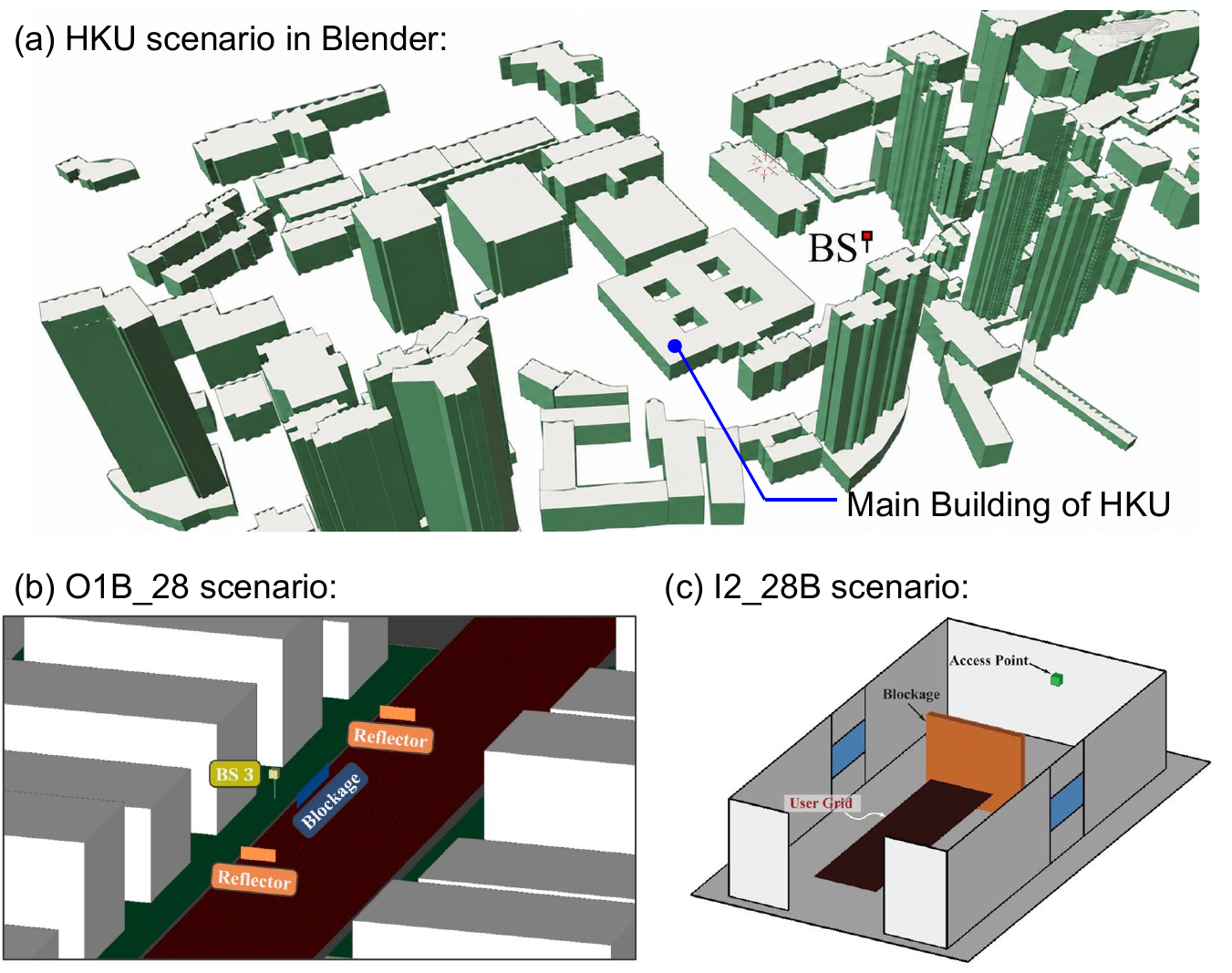}
	\caption{Illustration of the scenarios, (a) HKU\_28, (b) O1B\_28, and (c) I2\_28B.}
	\label{fig:scenarios}
\end{figure}

\subsection{Simulation Setup}
For each scenario, up to 100k user samples are randomly selected, with 80k used for training and the remainder for testing. As introduced in Section \ref{section3.B}, the velocity predictor is implemented by a customized 1D-DiT with embedding dimension 512, 5 transformer blocks, and 16 attention heads, and is trained for 300 epochs using AdamW with learning rate $2\times 10^{-4}$, weight decay 0.1, and batch size 32. The training follows a two-stage strategy: the first 150 epochs learn the standard instantaneous velocity, and the remaining 150 epochs fine-tune the model under split-consistency objective in Section \ref{section3.A} with $p=0.7$. To support variable probing budgets, 50\% of the training samples in each batch use the full 64-dimensional prompt, while for the remaining, only $Q$ uniformly spaced RSRP measurements are retained and the others are masked out. Additionally, exponential moving average (EMA) is adopted to update the parameters of the model with decay coefficient 0.995. For each scenario, the model was trained on an NVIDIA GeForce RTX 4060 with 8 GB of video random access memory (VRAM) for approximately 11 hours. With $N_t=Q_{{\rm max}}=64$ and $T=1$, the model contains 33.53M trainable parameters and requires around 2.88 Giga floating-point operations (GFLOPs) to generate a single beam candidate. Since the $M$ candidate states are processed jointly while sharing the prompt and time embeddings, the total complexity of $T$-step beam generation is approximately $T(2.785M+0.092)$ GFLOPs. Moreover, the computational cost is independent of the probing budget $Q$ since stochastic masking preserves the fixed $Q_{\rm max}$-dimensional condition representation. The detailed parameters are summarized in Table~\ref{table1}.

\begin{table}[t]
	\caption{Simulation Parameters}
	\label{table1}
	\centering
	\begin{tabular}{|c|c|}
		\hline
		Name of scenario & HKU\_28, O1B\_28, I2\_28B\\
		Carrier frequency $f_c$ & 28 GHz \\
		Bandwidth BW & 100 MHz \\
		Number of antennas $N_t$ & 64 \\
		Antenna array & ULA \\
		Antenna spacing & Half-wavelength spacing \\
		UE antenna & Single \\
		Antenna element & Isotropic \\
		Number of paths & 5 \\
		Learning rate & $2\times 10^{-4}$ \\
		AdamW Weight decay & 0.1 \\
		Batch size $B$ & 32 \\
		Number of epochs $E$ & 300 \\
		Probing budget set $\mathcal{Q}_p$ & \{9, 10, $\cdots$, 63, 64\} \\
		Stage I epochs $E_1$ & 150 \\
		$p$ & 0.7 \\
		$p_{{\rm full}}$ & 0.5 \\
		EMA decay coefficient & 0.995 \\
		\hline
	\end{tabular}
\end{table}

\subsection{Intuition Behind the Generated Beams}
Before performing the quantitative evaluation, we first provide an intuitive inspection of the generated beam pattern. Fig. \ref{fig:beam_patterns} compares the optimal beams in the HKU\_28, O1B\_28, and I2\_28B scenarios, shown in Figs. \ref{fig:beam_patterns}(a), (b), and (c), respectively, with their FBBS-generated counterparts in the second column, where $Q=15$, $M=5$ and $T=3$. It is observed that FBBS effectively captures the dominant angular structures of the optimal beams in HKU\_28 and O1B\_28. More interestingly, although the generated beam in Fig. \ref{fig:beam_patterns}(f) differs visibly from the optimal pattern in Fig. \ref{fig:beam_patterns}(c), it still achieves a normalized beamforming gain of $-0.3986$ dB, corresponding to approximately $91.2\%$ of the optimal received power. This observation suggests that FBBS does not simply reproduce the optimal beam pattern point by point, but instead synthesizes a \textit{functionally near-optimal beam} under the guidance of RSRP prompts.

\begin{figure}[t]
	\centering
	\includegraphics[width=3.5in]{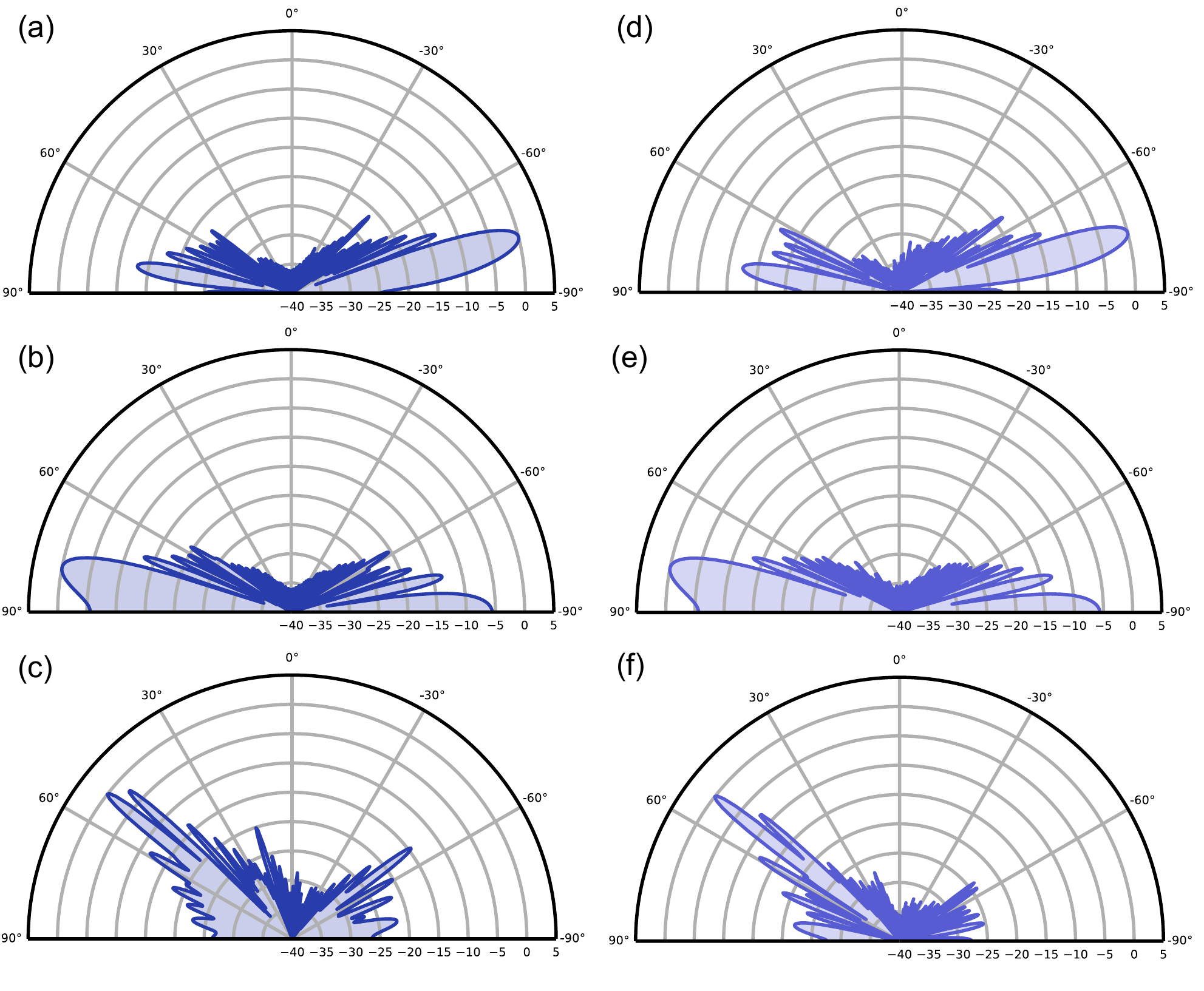}
	\caption{Beam patterns: optimal (a)-(c) versus generated (d)–(f).}
	\label{fig:beam_patterns}
\end{figure}

\subsection{Trade-off Between Gain and Probing Overhead}\label{sec:tradeoff_gain_probing_overhead}
\begin{figure}[t]
	\centering
	\includegraphics[width=3.5in]{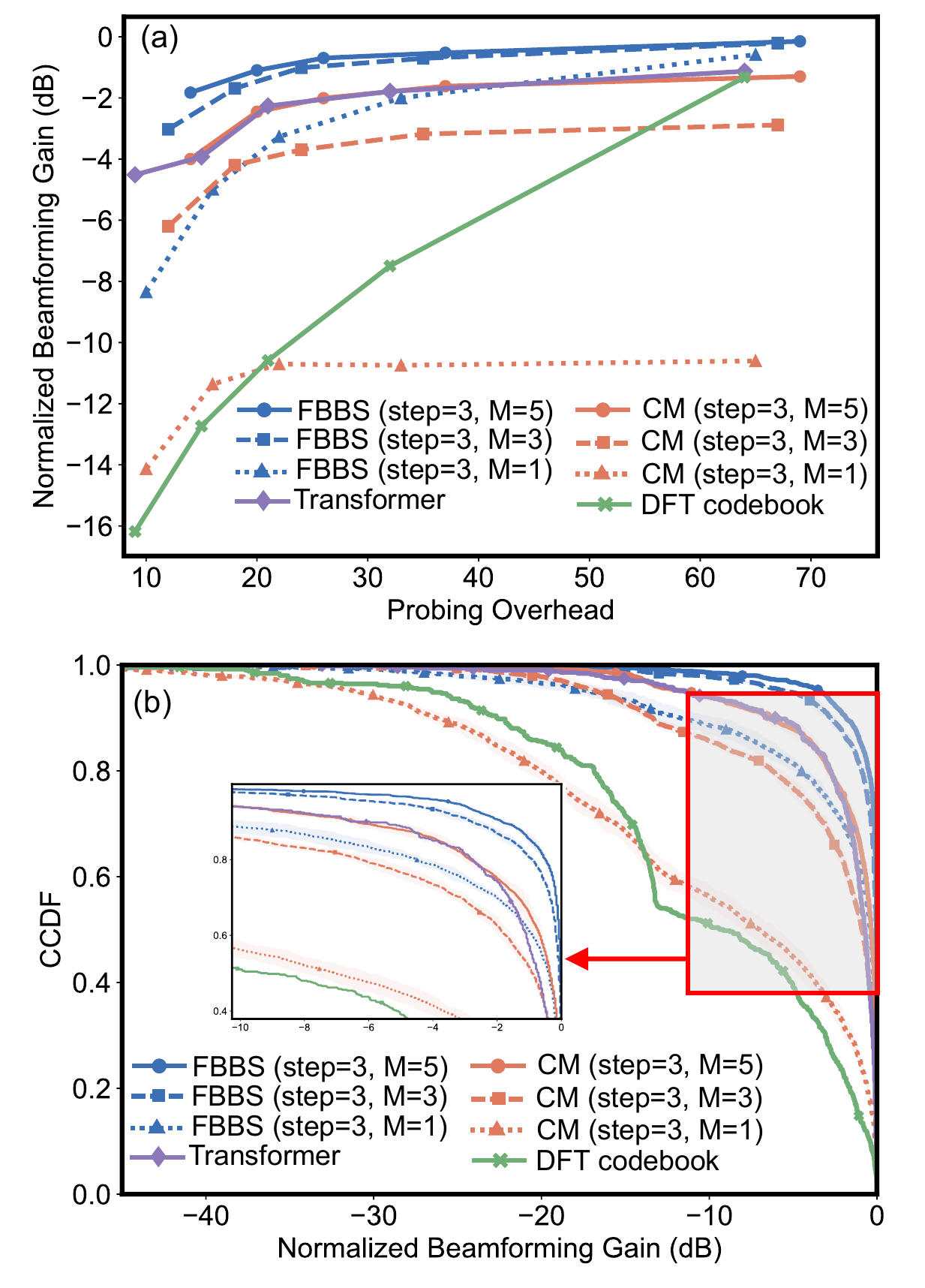}
	\caption{Beamforming performance in the HKU\_28 scenario: (a) average normalized beamforming gain versus the total probing overhead; (b) CCDF of the normalized beamforming gain with $Q=21$.}
	\label{fig:hku28_tradeoff_and_CCDF}
\end{figure}

\begin{figure}[t]
	\centering
	\includegraphics[width=3.5in]{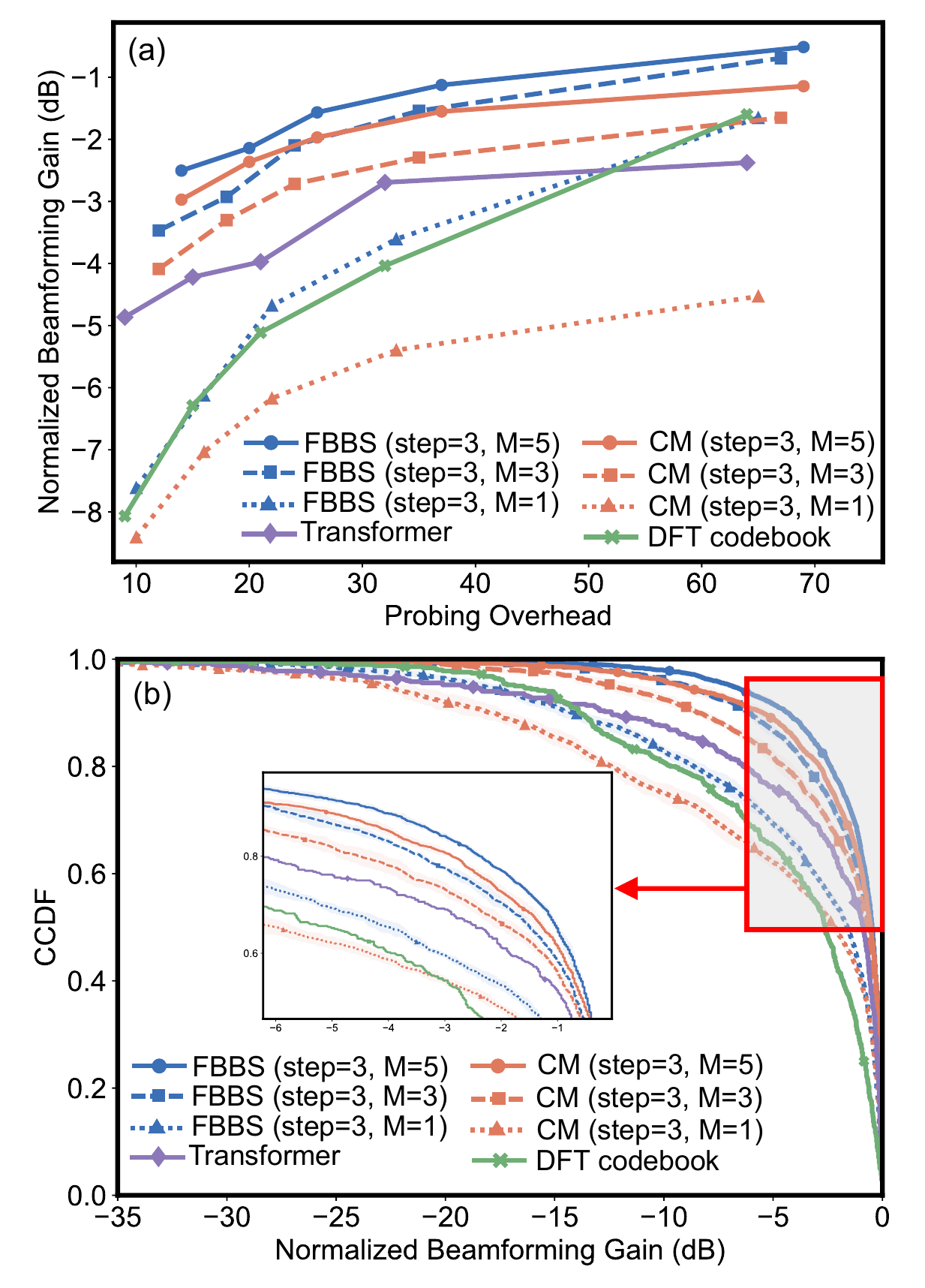}
	\caption{Beamforming performance in the O1B\_28 scenario: (a) average normalized beamforming gain versus the total probing overhead; (b) CCDF of the normalized beamforming gain with $Q=21$.}
	\label{fig:o1b28_tradeoff_and_CCDF}
\end{figure}

\begin{figure}[t]
	\centering
	\includegraphics[width=3.5in]{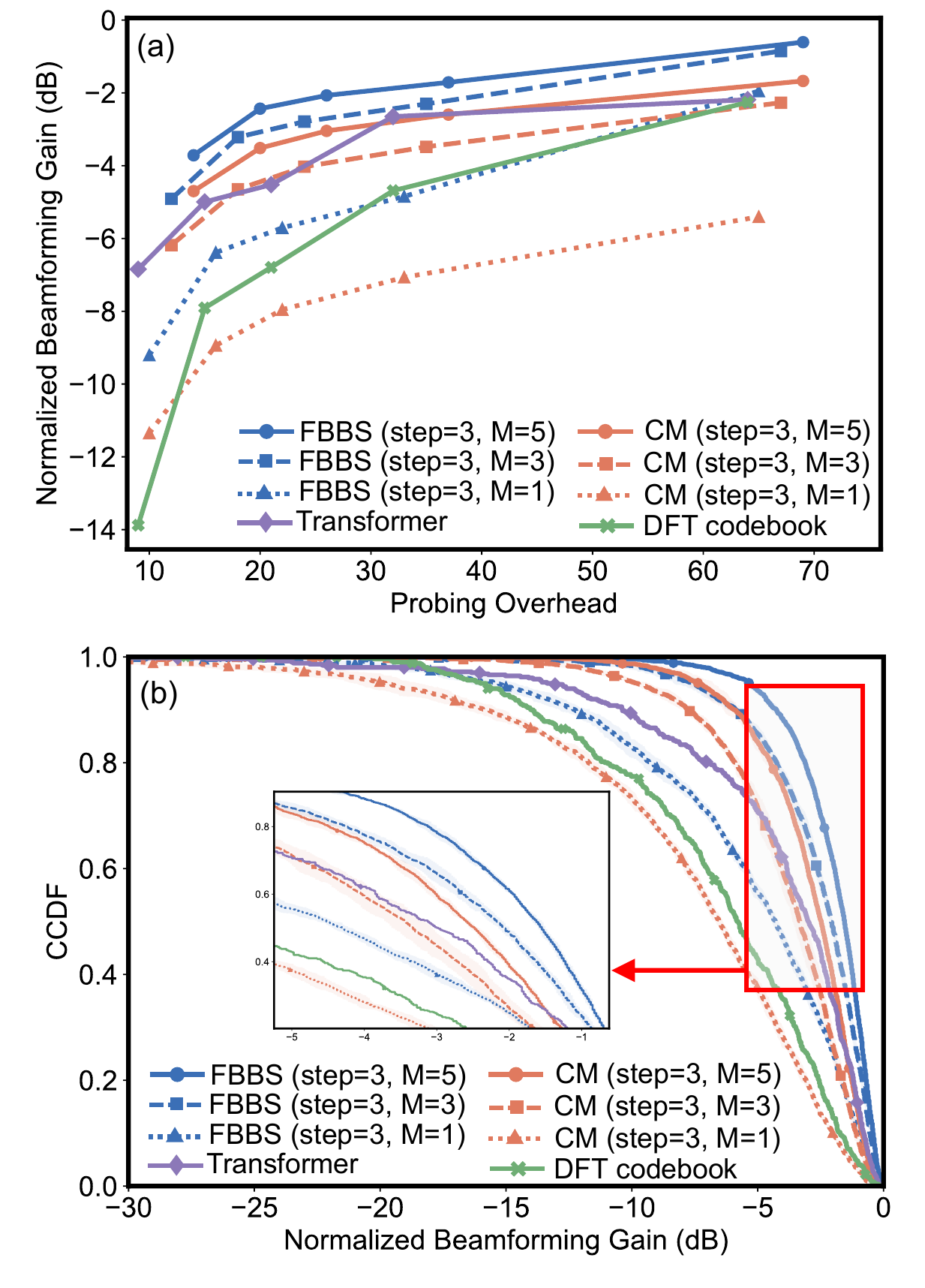}
	\caption{Beamforming performance in the I2\_28B scenario: (a) average normalized beamforming gain versus the total probing overhead; (b) CCDF of the normalized beamforming gain with $Q=21$.}
	\label{fig:i228b_tradeoff_and_CCDF}
\end{figure}

Qualitative analysis helps visually examine whether the proposed FBBS is capable of generating meaningful beams, whereas quantitative analysis is needed to more clearly reveal the actual performance gain it can achieve. In this work, the normalized beamforming gain is used as the main metric since it is positively associated with SNR. It can be calculated as
\begin{equation}
	\setlength\abovedisplayskip{3pt}
	\setlength\belowdisplayskip{3pt}
	\bar{g}_k(\bm{{\rm w}}) = 10{\rm log}_{10}\left(\frac{\vert \bm{{\rm h}}_k^{\rm H}\bm{{\rm w}}\vert^2}{\vert \bm{{\rm h}}_k^{\rm H}\bm{{\rm w}}_k^{\star}\vert^2}\right),
\end{equation}
where $\bm{{\rm w}}_k^{\star}$ denotes the optimal constant-modulus beamformer under the analog beamforming constraint in (\ref{eq:feasible_beam_set}). For comparison, we implement a consistency model-based GenSSBF, referred to as CM \cite{Song2023consistency}, and a Transformer-based beam predictor. Both baselines employ the same backbone as FBBS but were trained with different objectives. The upper panels of Figs. \ref{fig:hku28_tradeoff_and_CCDF} to \ref{fig:i228b_tradeoff_and_CCDF} demonstrate that FBBS consistently provides the most favorable tradeoff between beamforming gain and probing overhead across the considered scenarios. Its advantage is particularly pronounced in the low-overhead regime, where sparse RSRP measurements provide insufficient angular resolution for conventional codebook search but can still serve as informative prompts for generating user-specific beams. For example, in HKU\_28, with $Q=15, M=5$, FBBS achieves a normalized beamforming gain of -1.10 dB with a total overhead of $Q+M=20$, under three generation steps. Under the same overhead and configuration with FBBS, CM-based GenSSBF yields only -2.46 dB, and Transformer-based discriminative beam prediction method delivers -2.26 dB at $Q=21$. In contrast, the 32-beam DFT search attains merely -7.5 dB, despite consuming 32 probing beams. These results illustrate that with average-velocity field learning, FBBS converts a compact RSRP prompt into a high-quality user-specific beam more effectively. Moreover, with the proposed condition masking mechanism, a single FBBS model can flexibly adapt to different probing budgets without retraining. From Figs. \ref{fig:hku28_tradeoff_and_CCDF}(a), \ref{fig:o1b28_tradeoff_and_CCDF}(a), and  \ref{fig:i228b_tradeoff_and_CCDF}(a), we also notice that \textit{mild brainstorming} (e.g., $M$ from 1 to 3) is already sufficient to harvest most of the performance improvement, indicating that FBBS provides a practical knob to balance latency and beamforming accuracy: \textit{a small $M$ is suitable for fast access, whereas a moderately larger $M$ is beneficial when the system can afford slightly more overhead to pursue higher gain}.

The lower panels of Figs. \ref{fig:hku28_tradeoff_and_CCDF} to \ref{fig:i228b_tradeoff_and_CCDF} present the complementary cumulative distribution function (CCDF) of the normalized beamforming gain, where the $y$-axis represents the probability that a randomly selected UE achieves a gain higher than a given value on the $x$-axis. Therefore, a curve located farther toward the upper-right indicates that a larger proportion of UEs attain high beamforming gains. At $Q=21$, FBBS with $M=5$ consistently exhibits the most favorable CCDF across all three scenarios, particularly in the high-gain region highlighted by the insets. Increasing $M$ also produces a clear rightward shift, demonstrating that brainstorming improves not only the average gain but also user-level reliability by reducing the likelihood of poor beamforming outcomes.

\begin{table}[t]
	\centering
	\caption{Zero-shot cross-site performance of FBBS trained only 
		on HKU\_28 and directly evaluated on unseen scenarios without 
		fine-tuning, where $T=3$.}
	\label{tab:cross_site_generalization}
	\scriptsize
	\setlength{\tabcolsep}{3.2pt}
	\renewcommand{\arraystretch}{1.12}
	\begin{tabular}{ccrrrrr}
		\toprule
		\multirow{2}{*}{$M$}
		& \multirow{2}{*}{\textbf{Test Site}}
		& \multicolumn{5}{c}{\textbf{Probing Budget $Q$}} \\
		\cmidrule(lr){3-7}
		& & \textbf{9} & \textbf{15} & \textbf{21}
		& \textbf{32} & \textbf{64} \\
		\midrule
		\multirow{2}{*}{5}
		& O1B\_28
		& $-4.526$ & $-3.746$ & $-3.103$ & $-2.581$ & $-1.593$ \\
		& I2\_28B
		& $-7.273$ & $-3.737$ & $-3.582$ & $-2.948$ & $-2.377$ \\
		\midrule
		\multirow{2}{*}{8}
		& O1B\_28
		& $-3.730$ & $-3.162$ & $-2.689$ & $-2.170$ & $-1.369$ \\
		& I2\_28B
		& $-5.849$ & $-3.083$ & $-3.073$ & $-2.547$ & $-2.105$ \\
		\bottomrule
	\end{tabular}
\end{table}

Moreover, we evaluate the zero-shot cross-site transferability of FBBS, as shown in Table~\ref{tab:cross_site_generalization}. The model trained only on HKU\_28 is directly applied to O1B\_28 and I2\_28B without retraining or fine-tuning. Despite the substantial changes in site geometry and propagation characteristics, FBBS retains meaningful beamforming performance in both unseen environments. In particular, it achieves normalized gains of $-1.369$ dB in O1B\_28 and $-2.105$ dB in the fully NLoS I2\_28B scenario with $Q=64, M=8$. These results suggest that FBBS does not merely memorize the site-specific user distribution, but learns the underlying relationships between RSRP observations and beam structures.

\begin{figure}[t]
	\centering
	\includegraphics[width=3.5in]{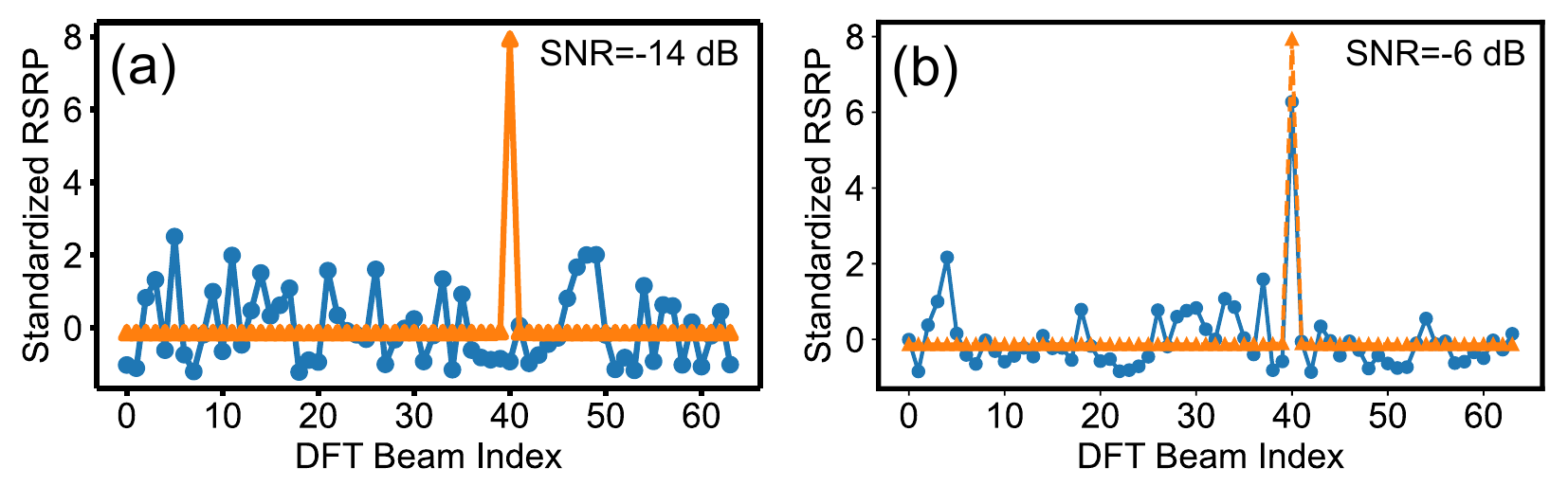}
	\caption{Noisy RSRP conditions in HKU\_28.}
	\label{fig:noisy_condition}
\end{figure}

\begin{figure}[h]
	\centering
	\includegraphics[width=3.5in]{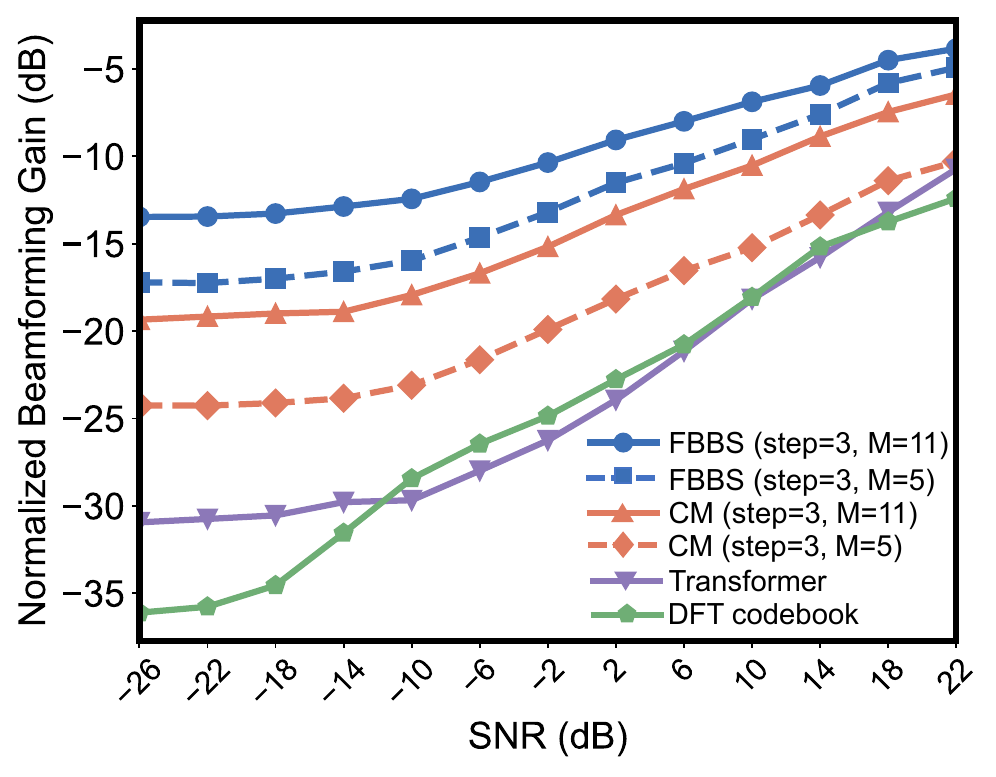}
	\caption{Normalized beamforming gain versus RSRP SNR in the HKU\_28 scenario, where $Q=32$.}
	\label{fig:hku28_noisy_conditions}
\end{figure}

\begin{figure}[t]
	\centering
	\includegraphics[width=3.5in]{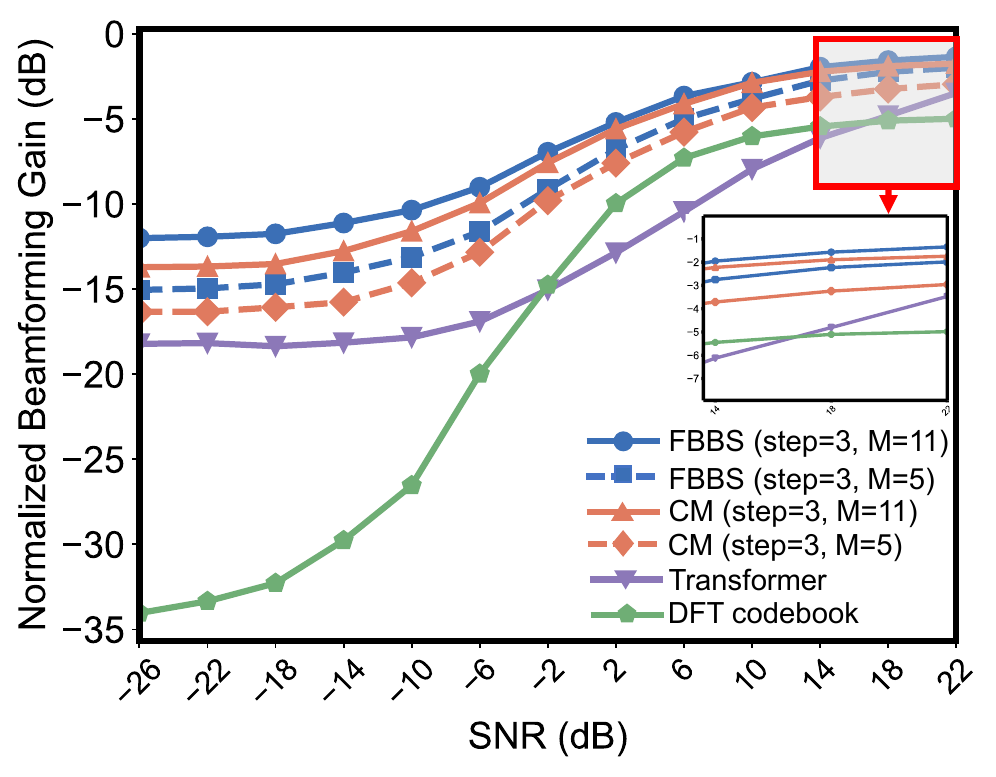}
	\caption{Normalized beamforming gain versus RSRP SNR in the O1B\_28 scenario, where $Q=32$.}
	\label{fig:o1b28_noisy_conditions}
\end{figure}

\begin{figure}[t]
	\centering
	\includegraphics[width=3.5in]{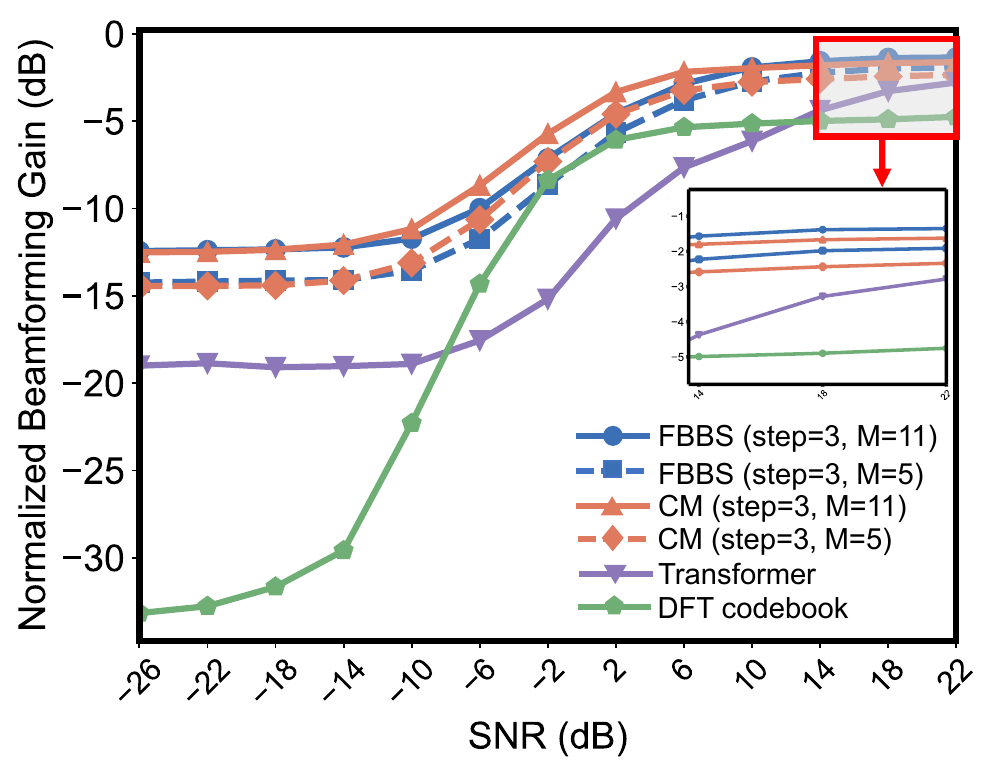}
	\caption{Normalized beamforming gain versus RSRP SNR in the I2\_28B scenario, where $Q=32$.}
	\label{fig:i228b_noisy_conditions}
\end{figure}

\subsection{The Impact of Noisy Wireless Prompts}\label{sec:noisy_prompts}
To evaluate the practicability of the proposed FBBS, it is important to test the performance under noisy environments. First, Fig. \ref{fig:noisy_condition} visualizes how noise affects the RSRP measurements. Figs. \ref{fig:hku28_noisy_conditions}--\ref{fig:i228b_noisy_conditions} show the normalized beamforming gain performance under different SNR levels. In this subsection, we set $Q=32$ and $T=3$ for FBBS and CM-based GenSSBF. It can be observed that although trained using noise-free wireless prompts, FBBS exhibits strong resilience to noisy conditions. Across all three scenarios, both FBBS configurations ($M\in \{5, 11\}$) consistently outperform the 32-beam Transformer-based beam prediction and DFT codebook search baselines. This advantage is particularly pronounced in the low-SNR regime, where direct codebook sweeping becomes highly susceptible to noise-induced beam-ranking errors, whereas FBBS can still exploit the learned site-specific channel distribution when the prompt contains limited information. The deterministic Transformer is also more vulnerable to this distribution shift because it produces only a single prediction from the corrupted input. Increasing the brainstorm number from $M=5$ to $M=11$ provides larger gain at low SNR: when SNR is -26 dB, the improvement ranges from approximately 1.8 to 3.7 dB across the three scenarios, whereas it decreases to only 0.6--1.1 dB when SNR is 22 dB. This observation indicates that candidate diversity can serve as an uncertainty-compensation mechanism: \textit{additional candidates are most valuable when noise makes the prompt ambiguous, while their marginal benefit diminishes once the prompt becomes sufficiently reliable.} However, it should be re-emphasized that this phenomenon only holds when the model has effectively learned the environmental information.

\subsection{The Impact of Outdated Prompts}\label{sec:outdated_prompts}
In this subsection, we investigate how the model reacts when the wireless prompt becomes outdated. In practice, RSRP measurements may be outdated due to UE mobility, channel fading, reporting delay and sudden blockage. A complex evaluation that considers both trajectory and geometry under continuous UE movement is beyond the scope of this paper and is left for future work. Instead, we conduct a preliminary controlled prompt aging test to assess whether the ``old'' condition remains effective in a new channel. Specifically, denote the channel of the $k$-th UE at time $t_0$ as $\bm{{\rm h}}_k(t_0)$. To
emulate channel variation, we construct a new channel as
\begin{equation}
	\setlength\abovedisplayskip{3pt}
	\setlength\belowdisplayskip{3pt}
	\bm{{\rm h}}_k^{\rm new} = \rho \bm{{\rm h}}_k(t_0) + \sqrt{1-\rho^2} \bm{{\rm e}}_k,
\end{equation}
where $\rho\in [0, 1]$ controls the correlation between $\bm{{\rm h}}_k^{\rm new}$ and $\bm{{\rm h}}_k(t_0)$. Smaller $\rho$ corresponds to more severe prompt aging. $\bm{{\rm e}}_k$  is an independent random perturbation with the same average power as $\bm{{\rm h}}_k(t_0)$. We keep the RSRP measurements obtained from $\bm{{\rm h}}_k(t_0)$, and the generated beams are evaluated on $\bm{{\rm h}}_k^{\rm new}$. In the simulation, we evaluate the generated beams under different channel correlation levels $\rho$.

\begin{figure*}[t]
	\centering
	\includegraphics[width=7.2in]{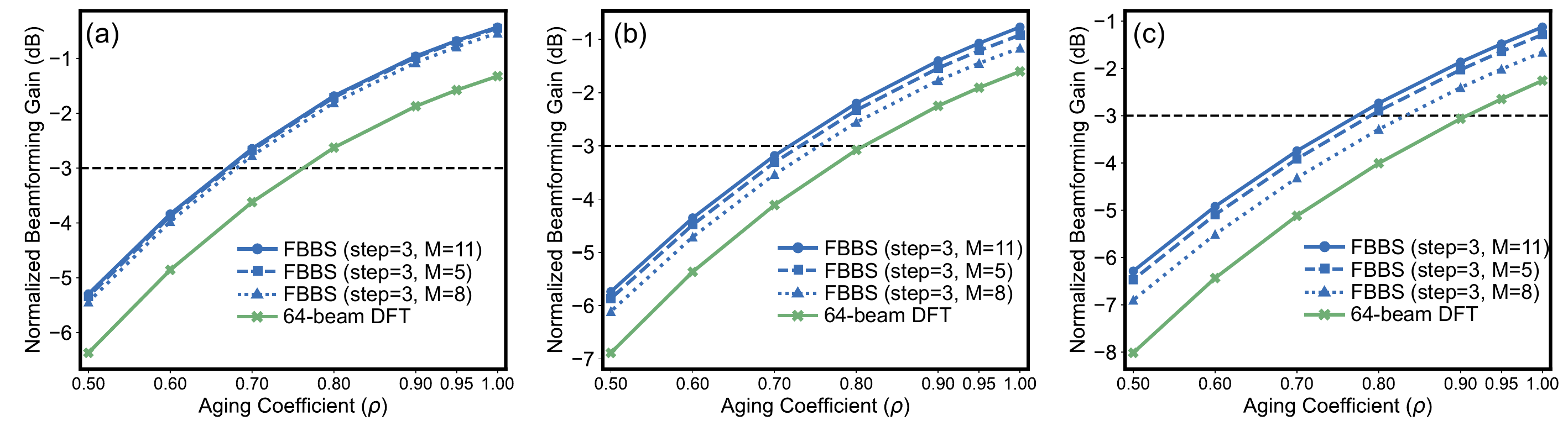}
	\caption{Normalized beamforming gain under different channel correlation levels in (a) HKU\_28, (b) O1B\_28, and (c) I2\_28B.}
	\label{fig:prompt_aging}
\end{figure*}

As shown in Fig. \ref{fig:prompt_aging}, the normalized beamforming gain decreases monotonically as $\rho$ decreases, confirming that increasingly outdated RSRP measurements provide less reliable guidance for beam generation. Nevertheless, FBBS consistently outperforms the 64-beam DFT baseline across all tested correlation levels and scenarios. At $\rho=0.5$, FBBS with $Q=21, M=11$ retains gains of approximately 1.08, 1.15, and 1.73 dB over DFT in HKU\_28, O1B\_28, and I2\_28B, respectively. With $-3$ dB as a reference, FBBS remains above this level down to $\rho=0.7$ in HKU\_28 and $\rho=0.8$ in O1B\_28 and I2\_28B, whereas DFT falls below it at $\rho=0.7$, $0.8$, and $0.9$, respectively. However, FBBS still loses approximately $5$ dB as $\rho$ decreases from 1 to 0.5. Therefore, designing a conditioning mechanism that is robust to outdated wireless prompts also represents an important research direction.

\subsection{Trade-off Between Gain and Generation Overhead}
In addition to the probing overhead investigated in Section~\ref{sec:tradeoff_gain_probing_overhead}, the online practicality of GenSSBF also depends critically on beam generation overhead, e.g., inference latency. For fixed antenna size and prompt dimensions, the generation overhead is essentially independent of the scenario and is determined primarily by the inference paradigm, the number of generation steps, and the backbone model. We therefore select O1B\_28 as a representative scenario and set $Q=21$ and $M=1$. As reported in Table~\ref{tab:gain_latency_tradeoff}, the proposed FBBS achieves a normalized beamforming gain of $-4.337$ dB with only one generation step and an average latency of $5.322$ ms. Its performance remains stable over $T=1$--$4$, with a variation of at most $0.301$ dB, indicating that FBBS has the capability to generate a high-quality beam in a single step. In contrast, standard flow matching (FM)\footnote{The FM and CM baselines employ the same customized 1D-DiT backbone and condition encoder as FBBS, differing only in their training objectives and corresponding sampling procedures.} \cite{lipman2022flow} improves markedly with $T$ but reaches only $-5.557$ dB with four steps and requires $23.256$ ms. Thus, one-step FBBS provides a $1.220$ dB higher gain while reducing latency by $77.1\%$. CM also demonstrates the ability of single-step generation, but its beamforming gain remains lower than that of FBBS. Even with $1000$ steps, diffusion-based BBS achieves only $-5.681$ dB and incurs $6.206$ s of latency, over three orders of magnitude higher than one-step FBBS. These results show that learning average velocity substantially reduces the reliance on iterative beam refinement, enabling FBBS to achieve the most favorable beamforming gain--generation overhead tradeoff.

\begin{table*}[t]
	\centering
	\caption{Beamforming gain--generation overhead trade-off of different
		GenSSBF methods (O1B\_28, $Q=21, M=1$). Each entry is reported as $(T,\bar{g},\tau)$, where $T$ denotes the number of generation steps, $\bar{g}$ is the normalized beamforming gain in dB, and $\tau$ is the average generation latency.}
	\label{tab:gain_latency_tradeoff}
	\setlength{\tabcolsep}{8pt}
	\renewcommand{\arraystretch}{1.15}
	\begin{tabular}{lcccc}
		\toprule
		FBBS
		& $(1,\,-4.337,\,5.322~\mathrm{ms})$
		& $(2,\,-4.507,\,11.195~\mathrm{ms})$
		& $(3,\,-4.469,\,18.809~\mathrm{ms})$
		& $(4,\,-4.638,\,23.256~\mathrm{ms})$ \\
		
		FM
		& $(1,\,-16.523,\,5.322~\mathrm{ms})$
		& $(2,\,-7.696,\,11.195~\mathrm{ms})$
		& $(3,\,-6.033,\,18.809~\mathrm{ms})$
		& $(4,\,-5.557,\,23.256~\mathrm{ms})$ \\
		
		CM
		& $(1,\,-6.267,\,5.322~\mathrm{ms})$
		& $(2,\,-6.271,\,11.195~\mathrm{ms})$
		& $(3,\,-6.352,\,18.809~\mathrm{ms})$
		& $(4,\,-6.321,\,23.256~\mathrm{ms})$ \\
		
		BBS
		& $(100,\,-15.786,\,634.742~\mathrm{ms})$
		& $(300,\,-9.964,\,1993.101~\mathrm{ms})$
		& $(500,\,-8.604,\,3078.699~\mathrm{ms})$
		& $(1000,\,-5.681,\,6206.203~\mathrm{ms})$ \\
		\bottomrule
	\end{tabular}
\end{table*}

\subsection{Discussion}\label{sec:discussion}
Although this work focuses on a ULA and single UE, the proposed average velocity-guided beam generation and flexible prompting are not inherently restricted to these settings. Extension to UPA or multi-antenna UE mainly requires adapting the beam-state representation, condition encoding, and feasibility mapping. Multi-user or hybrid beamforming further requires joint precoding states, UE-specific conditions, interference-aware objectives, and analog--digital hardware constraints. These extensions require dedicated training and validation and are left for future work.

\section{Conclusion}\label{Section5}
This paper proposed FBBS for fast GenSSBF under flexible probing budgets. By combining average-velocity learning with stochastic masking, a single model supports one- or few-step beam generation from variable-length RSRP observations and achieves a favorable gain--overhead--latency tradeoff, with mild brainstorming capturing most of the achievable gain. The results further demonstrate strong robustness to noisy and outdated prompts and meaningful zero-shot transfer to unseen sites. Future work will extend the framework to multi-user, multi-cell, and hybrid beamforming settings.

\newpage

\vspace{11pt}

\vfill

\end{document}